\documentclass[3p,preprint]{elsarticle}

\journal{Information and Computation}

\usepackage{amssymb}

\usepackage[figuresright]{rotating}

\newtheorem{theorem}{Theorem}
\newtheorem{lemma}[theorem]{Lemma}
\newtheorem{proposition}[theorem]{Proposition}
\newtheorem{corollary}[theorem]{Corollary}
\newtheorem{definition}[theorem]{Definition}
\newtheorem{claim}{Claim}
\newtheorem{example}{Example}

\newproof{proof}{Proof}

\usepackage[utf8]{inputenc}
\usepackage{complexity}
\usepackage[algosection]{algorithm2e}
\usepackage{endnotes}
\usepackage{listings}
\usepackage{mathtools}
\usepackage{verbatim}
\usepackage{mathrsfs}
\usepackage{shuffle}
\usepackage{longtable}
\usepackage{enumitem}
\usepackage{etoolbox}
\usepackage{float}
\usepackage{color}
\usepackage{textcomp}
\usepackage{hyperref}

\usepackage{mathtools}

\newcommand*{\DEBUG}{}%

\ifdefined\DEBUG
\newcommand{\fixme}[1]{{\textcolor{red}{\bf{\textsf{FIXME: #1}}}}}
\newcommand{\bug}[1]{{\textcolor{blue}{\bf{\textsf{BUG: #1}}}}}
\newcommand{\idea}[1]{{\textcolor{blue}{\bf{\textsf{IDEA: #1}}}}}

\newcommand{\TODO}[1]{{\textcolor{red}{\bf{\textsf{
TODO: #1
}}}}}
\else
\newcommand{\fixme}[1]{}
\newcommand{\bug}[1]{}
\newcommand{\TODO}[1]{}
\newcommand{\idea}[1]{}
\fi

\newclass{\DCM}{DCM}
\newclass{\DCA}{DCA}
\newclass{\DPDA}{DPDA}
\newclass{\NPDA}{NPDA}
\newclass{\DCMNE}{DCM_{NE}}
\newclass{\TwoDCM}{2DCM}
\newclass{\NCM}{NCM}
\newclass{\DPCM}{DPCM}
\newclass{\NPCM}{NPCM}
\newclass{\TRE}{TRE}
\newclass{\NFA}{NFA}

\newclass{\MPCA}{MPCA}
\newclass{\TCA}{TCA}
\newclass{\QCA}{QCA}
\newclass{\EPDA}{EPDA}

\newclass{\LL}{{\cal L}}

\newcommand\abs[1]{\left|#1\right|}

\newcommand\set[1]{\left\{#1\right\}}

\newcommand\rquot[2]{#1 {#2}^{-1}}

\newcommand\lquot[2]{#1^{-1} #2}

\newcommand{\sst}{\ensuremath{\mid}}

\newcommand\natnum{\mathbb{N}}
\newcommand\natzero{\mathbb{N}_0}

\DeclareMathOperator{\pref}{pref}
\DeclareMathOperator{\suff}{suff}
\DeclareMathOperator{\infx}{inf}
\DeclareMathOperator{\outf}{outf}
\DeclareMathOperator{\emb}{emb}

\newcommand\up{_{\uparrow}}
\newcommand\down{_{\downarrow}}

\begin{document}

\begin{frontmatter}




\title{Deletion Operations on Deterministic Families of Automata\tnoteref{t1}\tnoteref{t2}}

\tnotetext[t1]{This is an extended version of the conference paper \cite{TAMC2015}.}

\tnotetext[t2]{\textcopyright 2016. This manuscript version is made available under the CC-BY-NC-ND 4.0 license \url{http://creativecommons.org/licenses/by-nc-nd/4.0/}}


\author[label1]{Joey Eremondi\fnref{fn3}}
\address[label1]{Department of Information and Computing Sciences\\ Utrecht University, P.O.\ Box 80.089 3508 TB Utrecht, The Netherlands}
\ead[label1]{j.s.eremondi@students.uu.nl}

\author[label2]{Oscar H. Ibarra\fnref{fn2}}
\address[label2]{Department of Computer Science\\ University of California, Santa Barbara, CA 93106, USA}
\ead[label2]{ibarra@cs.ucsb.edu}
\fntext[fn2]{Supported, in part, by
NSF Grant CCF-1117708 (Oscar H. Ibarra).}

\author[label3]{Ian McQuillan\corref{corr1}\fnref{fn3}}
\address[label3]{Department of Computer Science, University of Saskatchewan\\
Saskatoon, SK S7N 5A9, Canada}
\ead[label3]{mcquillan@cs.usask.ca}
\cortext[corr1]{Corresponding author}
\fntext[fn3]{Supported, in part, by Natural Sciences and Engineering Research Council of Canada Grant 327486-2010 (Ian McQuillan).}

\begin{abstract}
Many different deletion operations are investigated applied to languages accepted by one-way and two-way deterministic reversal-bounded multicounter machines, deterministic pushdown automata, and finite automata. Operations studied include the prefix, suffix, infix and outfix operations, as well as left and right quotient with languages from different families. It is often expected that language families defined from deterministic machines will not be closed under deletion operations. However, here, it is shown that one-way deterministic reversal-bounded multicounter languages are closed under right quotient with languages from many different language families; even those defined by nondeterministic machines such as the context-free languages. Also, it is shown that when starting with one-way deterministic machines with one counter that makes only one reversal, taking the left quotient with languages from many different language families --- again including those defined by nondeterministic machines such as the context-free languages --- yields only one-way deterministic reversal-bounded multicounter languages (by increasing the number of counters). These results are surprising given the nondeterministic nature of the deletion operation. However, if there are two more reversals on the counter, or a second 1-reversal-bounded counter, taking the left quotient (or even just the suffix operation) yields languages that can neither be accepted by deterministic reversal-bounded multicounter machines, nor by 2-way nondeterministic machines with one reversal-bounded counter. 
\end{abstract}

\begin{keyword}
Automata and Logic \sep Counter Machines \sep Deletion Operations \sep Reversal-Bounds \sep Determinism \sep Finite Automata
\end{keyword}

\end{frontmatter}

\section{Introduction}
\label{sec:intro}

This paper involves the study of various types of deletion operations applied to languages accepted by deterministic 
classes of machines. Deletion operations, such as left and right quotients, and word operations such as prefix, suffix, infix,
and outfix, are more commonly studied applied to languages accepted by classes of nondeterministic machines. Indeed, many language families accepted
by nondeterministic acceptors form {\em full trios} (closure under homomorphism, inverse homomorphism, and intersection with 
regular languages), and every full trio is closed under left and right quotient with regular languages, prefix, suffix,
infix, and outfix \cite{G75}. For families of languages accepted by deterministic machines however, the situation is more
tricky due to the nondeterministic behaviour of the deletion. 
Indeed, deterministic pushdown automata are not even closed under left quotient with a set of individual letters. Here, most deterministic machine models studied will involve restrictions of one-way deterministic
reversal-bounded multicounter machines $(\DCM)$. These are machines that operate like deterministic finite automata with an additional fixed number of counters, where there is a bound on the number of times each counter switches between increasing and decreasing
\cite{Baker1974,Ibarra1978}. 
The family $\DCM(k,l)$ consists of languages accepted by
machines with $k$ counters that are $l$-reversal-bounded.
$\DCM$ languages have many decidable properties, such as emptiness, infiniteness, equivalence, inclusion, universe, and disjointness \cite{Ibarra1978}. Furthermore, $\DCM(1,l)$ forms an important 
restriction of deterministic pushdown automata.

These machines have been studied in a variety of different applications, such as to membrane computing \cite{counterMembrane}, verification of infinite-state systems \cite{verification,stringTransducers,modelChecking,verificationDiophantine}, and Diophantine equations \cite{verificationDiophantine}.

Recently, in \cite{EIMInsertion2015}, a related study was conducted for insertion operations; specifically operations defined by ideals obtained from the prefix, suffix, infix, and outfix relations, as well as left and right concatenation with languages from different language families. It was found that languages accepted by one-way deterministic reversal-bounded counter machines with one reversal-bounded counter are closed under right concatenation with $\Sigma^*$, but having two 1-reversal-bounded counters and right concatenating $\Sigma^*$ yields languages outside of both $\DCM$ and $2\DCM(1)$ (languages accepted by two-way deterministic machines with one counter that is reversal-bounded). It also follows from this analysis that the right input end-marker is necessary for even one-way deterministic reversal-bounded counter machines, when there are at least two counters. Furthermore, concatenating $\Sigma^*$ to the left of some one-way deterministic 1-reversal-bounded one counter languages yields languages that are neither in $\DCM$ nor $2\DCM(1)$. 
Other recent results on reversal-bounded multicounter languages include a technique to show languages are outside of $\DCM$ \cite{Chiniforooshan2012}. 
%
Closure properties of nondeterministic counter machines
under other types of deletion operations were studied in \cite{parInsDel}.

In this paper we investigate closure properties of types of deterministic machines. In Section \ref{sec:prelims}, preliminary background and notation
are introduced. In Section \ref{sec:closure}, erasing operations where $\DCM$ is closed are studied. It is shown that $\DCM$ is closed under
right quotient
with context-free languages, and that the left quotient of $\DCM(1,1)$
by a context-free language is in $\DCM$. Both results are generalizable to quotients
with a variety of different families of languages containing only semilinear languages. In Section \ref{sec:nonclosure}, non-closure of $\DCM$
under erasing operations are studied. It is shown that the set of suffixes,
infixes, or outfixes of a $\DCM(1,3)$ or $\DCM(2,1)$ language can
be outside of both $\DCM$ and $2\DCM(1)$.
In Section \ref{sec:DPDA}, $\DPCM$s (deterministic pushdown automata augmented by reversal-bounded
counters), and $\NPCM$s (the nondeterministic variant) are studied.
It is shown that $\DPCM$ is not closed under prefix or suffix,
and the right or left quotient of the language accepted by a $1$-reversal-bounded deterministic pushdown automaton
by a $\DCM(1,1)$ language can be outside $\DPCM$. In Section \ref{sec:reg},
the effective closure of regular languages with other families is briefly
discussed, and in Section \ref{sec:bounded}, bounded languages are discussed.

\section{Preliminaries}
\label{sec:prelims}

The set of non-negative integers is denoted by $\natzero$, and the set of positive integers by $\natnum$. For $c \in \natzero$,
let $\pi(c)$ be $0$ if $c=0$, and $1$ otherwise. 

We assume knowledge of standard formal language theoretic concepts such as finite automata, determinism, nondeterminism, semilinearity, recursive, and recursively enumerable languages \cite{Baker1974,HU}.
Next, we will give some notation used in the paper.
The empty word is denoted by $\lambda$.
If $\Sigma$ is a finite alphabet, then
$\Sigma^*$ is the set of all words over $\Sigma$
and $\Sigma^+ = \Sigma^* \setminus \set{\lambda}$. 
For a word $w \in \Sigma^*$, if $w = a_1 \cdots a_n$
where $a_i \in \Sigma$, $1\leq i \leq n$, the length of $w$ is denoted by $\abs{w}=n$,
and the reversal of $w$ is denoted by $w^R = a_n \cdots a_1$,
which is extended to reversals of languages in the natural way. 
In addition, if $a \in \Sigma$, $|w|_a$ is the number of $a$'s in $w$.
A language over $\Sigma$ is any subset of $\Sigma^*$.
Given a language $L\subseteq \Sigma^*$,
the complement of $L$, $\Sigma^* \setminus L$ is denoted by $\overline{L}$.
Given two languages $L_1,L_2$, the left quotient of $L_2$ by $L_1$, $L_1^{-1}L_2 = \{ y \mid xy \in L_2, x \in L_1\}$, and the right quotient of $L_1$ by $L_2$ is $L_1 L_2^{-1} = \{x \mid xy \in L_1, y \in L_2\}$.
A {\em full trio} is a language family closed under homomorphism,
inverse homomorphism, and intersection with regular languages \cite{HU}.

Let $n \in \mathbb{N}$. Then $Q \subseteq \mathbb{N}_0^n$ 
is a {\em linear} set if there is a vector $\vec{c} \in \mathbb{N}_0^n$
(the constant vector), and a set of vectors 
$V = \{\vec{v_1}, \ldots, \vec{v_r} \}, r \geq 0$, each 
$\vec{v_i} \in \mathbb{N}_0^n$ such that 
$Q = \{c + t_1 \vec{v_1} + \cdots + t_r\vec{v_r} \mid t_1, \ldots, t_r \in \mathbb{N}_0\}$. A finite union of linear sets is called a 
{\em semilinear set}.

A language $L$ is
{\em word-bounded} or simply {\em bounded}
if $L \subseteq w_1^* \cdots w_k^*$ for some
$k \ge 1$ and 
(not-necessarily distinct)
words $w_1, \ldots, w_k$.
Further, $L$ is {\em letter-bounded} if each $w_i$ is a letter.
Also, $L$ is {\em bounded-semilinear} if $L \subseteq
w_1^* \cdots w_k^*$ and $Q = \{(i_1, \ldots, i_k) ~|~
w_1^{i_1} \cdots w_k^{i_k} \in L \}$ is a semilinear set \cite{boundedSemilin}.


We now present notation for common word and language operations
used throughout the paper. 

\begin{definition}
\label{def:opGeneralize}
For a language $L \subseteq \Sigma^*$, 
the prefix, suffix, infix, and outfix operations are defined by:
\begin{itemize}
\item $\pref(L) = \set{w \sst wx \in L, x \in \Sigma^* }$,

\item $\suff(L) = \set{w \sst xw \in L, x \in \Sigma^* }$,

\item $\infx(L) = \set{w \sst xwy \in L, x,y \in \Sigma^* }$,

\item $\outf(L) = \set{xy \sst xwy \in L, w \in \Sigma^* }$.
\end{itemize}


%
\end{definition}

Note that $\pref(L) = L ( \Sigma^*)^{-1}$ and $\suff(L) = (\Sigma^*)^{-1}L$.

The outfix operation has been generalized to the notion of embedding \cite{JKT}:
\begin{definition}
The $m$-embedding of a language $L \subseteq \Sigma^*$ is the following set:
$\emb(L, m) = \{w_0 \cdots w_m \sst w_0 x_1 \cdots w_{m-1} x_m w_m \in L$,
$w_i \in \Sigma^*, 0 \leq i \leq m, x_j \in \Sigma^*, 1 \leq j \leq m\}$.
\end{definition}
Note that $\outf(L) = \emb(L, 1)$.

A {\em nondeterministic multicounter machine} is a finite automaton augmented by a fixed number of counters. The counters can be increased, decreased, tested for zero, or tested to see if the value is positive. A multicounter machine is {\em reversal-bounded} if there exists a fixed $l$ such that, in every accepting computation, the count on each counter alternates between increasing and decreasing at most $l$ times.

Formally, a {\em one-way $k$-counter machine} is a tuple $M = (k,Q,\Sigma, \lhd, \delta, q_0, F)$, where
$Q, \Sigma, \lhd, q_0,F$ are respectively the finite set of states, the input alphabet, the right input end-marker, the initial state in $Q$, and the set of final states that is a subset of $Q$. The transition function $\delta$ (defined as in \cite{Ibarra1978} except with only a right end-marker since we only use one-way inputs) is a relation from $Q \times (\Sigma \cup \{\lhd\}) \times \{0,1\}^k$ into $Q \times \{{\rm S},{\rm R}\} \times  \{-1, 0, +1\}^k$, such that if $\delta(q,a,c_1, \ldots, c_k)$ contains $(p,d,d_1, \ldots, d_k)$ and $c_i =0$ for some $i$, then $d_i \geq 0$ to prevent negative values in any counter. 
The direction of the input tape head movement is given by the symbols ${\rm S}$ and ${\rm R}$ for either {\em stay} or {\em right} respectively. The machine $M$
is {\em deterministic} if $\delta$ is a partial function.
A {\em configuration} of $M$ is a $k+2$-tuple $(q, w\lhd , c_1, \ldots, c_k)$ for describing the situation where
$M$ is in state $q$, with $w\in \Sigma^* $ still to read as input, and $c_1, \ldots, c_k\in \natzero$ are the
contents of the $k$ counters. The derivation relation $\vdash_M$ is defined between configurations, where $(q, aw, c_1, \ldots , c_k) \vdash_M (p, w'$ $, c_1 + d_1, \ldots, c_k+d_k)$, if
$(p, d, d_1, \ldots, d_k) \in \delta(q, a, \pi(c_1), \ldots, \pi(c_k))$ where $d \in \{{\rm S}, {\rm R}\}$ and 
$w' =aw$ if $d={\rm S}$, and $w' = w$ if $d={\rm R}$. Extended derivations are given by $\vdash^*_M$, the reflexive, transitive closure of $\vdash_M$.
A word $w\in \Sigma^*$ is accepted by $M$ if 
$(q_0, w\lhd, 0, \ldots, 0) \vdash_M^* (q, \lhd, c_1, \ldots, c_k)$, for some $q \in F$, and 
$c_1, \ldots, c_k \in \natzero$. The language accepted by $M$, denoted by $L(M)$, 
is the set of all words accepted by $M$. 
The machine $M$ is $l$-reversal bounded if, in every accepting computation, the count on each counter alternates between increasing and decreasing at most $l$ times.

We denote by $\NCM(k,l)$ the family of languages accepted by one-way 
nondeterministic $l$-reversal-bounded $k$-counter machines.
We denote by $\DCM(k,l)$ the family of languages accepted by
one-way deterministic $l$-reversal-bounded $k$-counter machines.
The union of the families of languages are denoted by
 $\NCM = \bigcup_{k,l \geq 0} \NCM(k,l)$ and
 $\DCM = \bigcup_{k,l \geq 0} \DCM(k,l)$. 
 Further, $\DCA$ is the family of languages accepted by
 one-way deterministic one counter machines (no reversal-bound).
 We will also sometimes refer to a multicounter machine as being in $\NCM(k,l)$ ($\DCM(k,l)$), if it has $k$ $l$-reversal bounded counters (and is deterministic).
Note that a counter that makes $l$ reversals
can be simulated by $\lceil \frac{l+1}{2} \rceil$ 1-reversal-bounded
counters. Hence, for example, $\DCM(1,3) \subseteq \DCM(2,1)$.

We denote by $\REG$ the family of regular languages, by $\NPDA$ the family
of context-free languages, by $\DPDA$ the family of deterministic
pushdown languages, by $\DPDA(l)$ the family of $l$-reversal-bounded
deterministic pushdown automata (with an upper bound of $l$ on the number
of changes between non-increasing and non-decreasing the size of the pushdown),
by $\NPCM$ the family of languages accepted by nondeterministic pushdown automata augmented by a fixed number of reversal-bounded counters \cite{Ibarra1978}, and by $\DPCM$ the deterministic variant. We also denote by $\TwoDCM$ the family of languages accepted by two-way input, deterministic  finite automata (both a left and right input tape end-marker are required) augmented by reversal-bounded counters, and by $2\DCM(1)$, $2\DCM$ with one reversal-bounded counter \cite{IbarraJiang}. A machine of this form is said to be {\em finite-crossing} if there is a fixed $c$ such that the number of times the boundary between any two adjacent input cells is crossed is at most $c$ \cite{Gurari1981220}. A machine is {\em finite-turn} if the input head makes at most $k$ turns on the input, for some $k$. Also, $2\NCM$ is the family of languages accepted by two-way nondeterministic machines with a fixed number of reversal-bounded counters, while $2\DPCM$ is the family of two-way deterministic pushdown machines augmented by a fixed number of reversal-bounded counters.

We give four examples below to illustrate the workings of the
reversal-bounded counter machines.

\begin{example}
Let $L = \{a^nb^n \mid n \ge 1\}$.  This language can be
accepted by a $\DCM(1,1)$ machine which reads $a^n$ and stores
$n$ in the counter and then reads the $b$'s while
decrementing the counter.  It accepts if the counter
becomes zero at the end of the input.
\end{example}

\begin{example}
Let $L_k = \{ x_1 \#  \cdots \# x_k \mid x_i \in \{a,b\}^+, x_j \ne x_k$
for $j \ne k \}$.  This can be  accepted by a  $\NCM(k(k+1)/2, 1)$, where we identify counter $i$ by the name $C_i$, for $ 1 \leq i \leq k(k+1)/2$.
The machine $M_k$ operates by reading the input and verifying
that for $1 \le i < j \le k$, $x_i$ and $x_j$ are different in
at least one position, or are different lengths, which is easy to
check using the counters.  To verify that they differ in at least one
position, while scanning $x_i$,
$M_k$ stores in counter $C_i$ a nondeterministically guessed 
position $p_i$ of $x_i$ and records in the state the 
symbol $a_{p_i}$ in that location.
Then, when scanning $x_j$, $M_k$
stores in counter $C_j$ a guessed location $p_j$ of $x_j$ and
records in the state the symbol $a_{p_j}$ in that location.  At the end
of the input, on stay transitions, $M_k$ checks that
$a_{p_i} \ne a_{p_j}$ and $p_i = p_j$ (by decrementing counters
$C_i$ and $C_j$ simultaneously and verifying that they become
zero at the same time).
\end{example}

\begin{example}
Let $L = \{x x^R \mid x \in \{a,b\}^+, |x|_a  > |x|_b \}$, which is not a
context-free language, can be accepted by an $\NPCM(2,1)$ $M$ (with counters labelled $C_1$ and $C_2$)
which operates as follows:  It scans the
input and uses the pushdown to check that the input is of the form
$x x^R$, by guessing the middle of the string, pushing $x$, and popping 
$x$ while reading $x^R$. In parallel, $M$ uses two counters $C_1$ and $C_2$ to store the numbers of
$a$'s and $b$'s it encounters. Then, at the end of the input,
on stay transitions, $M$ decrements $C_1$ and $C_2$ simultaneously
and verifies that the contents of $C_1$ is larger than that of $C_2$.
Clearly the counters are 1-reversal-bounded.

Note that the language $L' = \{x \#x^R \mid x \in \{a,b\}^+, |x|_a > |x|_b \}$, where $\#$ is a new symbol, is in $\DPCM(2,1)$.
\end{example}

\begin{example}
$L = \{a^mb^n \mid m, n \ge 1,  m \mbox{~is divisible by~} n\}$
can be accepted by a $2\DCM(1)$, which reads $a^m$
and stores $m$ in the counter.  Then it checks that
$m$ is divisible by $n$ by making left-to-right and
right-to-left sweeps on $b^n$ while decrementing
the counter.
\end{example}

We note that each of $\NCM(k,l), \NPCM(k,l), \NCM, \NPCM, \NPDA, \REG$ (for each $k,l$) form full trios (discussed in Section \ref{sec:intro}) 
\cite{Ibarra1978,G75}, and are therefore immediately closed under left and right quotient with regular languages, prefix, suffix, 
infix, and outfix. 

The next result proved in \cite{boundedSemilin} gives examples of weak and strong machines that are equivalent over word-bounded languages.
\begin{theorem} \cite{boundedSemilin} \label{WWW}
	The following are equivalent for every
	word-bounded language $L$:
	\begin{enumerate}
        \item $L$ can be accepted by an $\NCM$.
		\item $L$ can be accepted by an $\NPCM$.
        \item $L$ can be accepted by a finite-crossing $2\NCM$.
	\item $L$ can be accepted by a $\DCM$.
	\item $L$ can be accepted by a finite-turn $2\DCM(1)$.
	\item $L$ can be accepted by a finite-crossing $2\DPCM$
	\item $L$ is bounded-semilinear.
	\end{enumerate}
\end{theorem}

\noindent
We also need the following result in \cite{IbarraJiang}:
\begin{theorem} \cite{IbarraJiang} \label{unary}
Let $L \subseteq a^*$ be accepted by a
$2\NCM$ (not necessarily finite-crossing).
Then $L$ is regular, hence, semilinear.
\end{theorem}

\section{Closure of $\DCM$ Under Erasing Operations}
\label{sec:closure}

First, we discuss the left quotient of $\DCM$ with finite sets.
\begin{proposition} $\DCM$ is closed under left quotient with finite languages.
\end{proposition}
\begin{proof}
It is clear that $\DCM$ is closed under left quotient with a single word. Then the result follows from closure of $\DCM$ under union \cite{Ibarra1978}.
\qed
\end{proof}
This is in contrast to $\DPDA$, which is not even closed under left quotient with sets of multiple letters. Indeed, the language 
$\{\#a^n b^n \mid n > 0\} \cup \{\$ a^n b^{2n} \mid n> 0\}$ is a $\DPDA$ language, but taking the left quotient with $\{\$,\#\}$ produces a language which is not a $\DPDA$ language \cite{GinsburgDPDAs}.

Next, we show the closure of $\DCM$ under right quotient with
languages accepted by any 
nondeterministic reversal-bounded multicounter machine,
even when augmented with a pushdown store.

\begin{proposition}
\label{fullQuotientClosure}
Let $L_1 \in \DCM$ and let $L_2 \in \NPCM$.
Then $\rquot{L_1}{L_2} \in \DCM$. 
\end{proposition}
\begin{proof}
Consider a $\DCM$ machine $M_1=(k_1,Q_1, \Sigma, \lhd, \delta_1, s_0, F_1)$ and 
$\NPCM$ machine $M_2$ over $\Sigma$ with $k_2$ counters
where $L(M_1) = L_1$ and $L(M_2) = L_2$.
A $\DCM$ machine $M'$ will be constructed accepting $\rquot{L_1}{L_2}$.

Let $\Gamma = \{a_1, \ldots, a_{k_1}\}$ be new symbols.
First, intermediate $\NPCM$ machines will be built, one for each state $q$ of $M_1$. The intuition will be described first.
The machine built for state $q$ 
will accept all counter values (encoded over
$\Gamma$) of the form $a_1^{p_1} \cdots a_{k_1}^{p_{k_1}}$, whereby
there exists some $x \in \Sigma^*$, such that: \begin{itemize}
\item $M_1$ accepts $x$ starting from
state $q$ and counter values $(p_1, \ldots, p_{k_1})$, and 
\item $M_2$ accepts $x$.\end{itemize}
The key use of these machines is that they are all bounded languages, and
therefore Theorem \ref{WWW} applies, and they can all
be effectively converted to a $\DCM$ machine. From these, our final deterministic machine $M'$ can be built
by simulating $M_1$ until it hits the end-marker $\lhd$ in some state $q$. It then
deterministically simulates the {\it values that remain in the counters} with the
intermediate $\DCM$ machine constructed for state $q$, accepting if this intermediate
machine does. Indeed, it will therefore accept if and only if, from these counter values and state, $M_1$ can continue to lead to acceptance on some word $x$,
and $x$ is also in $L_2$. Hence, the deterministic simulation of
the intermediate machines is the key.
 
Formally, for each $q \in Q_1$, let $M_c(q)$ be an intermediate $k_1+k_2$ counter (plus a pushdown) $\NPCM$ machine over $\Gamma$ constructed as follows:
on input $a_1^{p_1} \cdots a_{k_1}^{p_{k_1}}$, $M_c(q)$ increments the first $k_1$ counters to $(p_1, \ldots, p_{k_1})$. Then $M_c(q)$ nondeterministically guesses a word $x\in \Sigma^*$ (by using only stay transitions of $M_c(q)$ but simulating lettered transitions on each letter of $x$ nondeterministically one at a time) and simulates $M_1$ on $x\lhd$ starting from state $q$ and from the counter values of $(p_1, \ldots, p_{k_1})$ using the first $k_1$ counters, while in parallel, simulating $M_2$ on $x$ using the next $k_2$ counters and the pushdown. This is akin to the product automaton construction described in \cite{Ibarra1978} showing $\NPCM$ is closed under intersection with $\NCM$. Then $M_c(q)$ accepts if both $M_1$ and $M_2$ accept.

\begin{claim}
Let $ L_c(q) = \{ a_1^{p_1} \cdots a_{k_1}^{p_{k_1}} \sst \exists x \in L_2 \mbox{~such that ~} (q, x\lhd, p_1, \ldots, p_{k_1}) \vdash_{M_1}^* (q_f, \lhd, p'_1, \ldots , p'_{k_1}) , p'_i \geq 0, 1 \leq i \leq k_1, q_f \in F_1 \}$.
Then $L(M_c(q)) = L_c(q)$.
\end{claim}
\begin{proof}
Consider $w = a_1^{p_1} \cdots a_{k_1}^{p_{k_1}} \in L_c(q)$.
Then there exists  $x$ where $x \in L_2$
and $(q, x\lhd, p_1, \ldots, p_{k_1}) \vdash_{M_1}^* (q^1_f, \lhd, p'_1, \ldots , p'_{k_1})$,
where $q^1_f \in F_1$.
There must then be some final state $q^2_f \in F_2$ reached
when reading $x\lhd$ in $M_2$.
Then, $M_c(q)$, on input $w$ places $(p_1, \ldots, p_{k_1}, 0, \ldots, 0)$ on the counters and then can nondeterministically guess $x$ letter-by-letter and
simulate $x$ in $M_1$ from state $q$ on the first $k_1$ counters and 
simulate $x$ in $M_2$ from its initial configuration on the remaining counters and pushdown. Then $M_c(q)$ 
ends up in state $(q^1_f, q_f^2)$, which is final.
Hence, $w \in L(M_c(q))$.

Consider $w = a_1^{p_1} \cdots a_{k_1}^{p_{k_1}} \in L(M_c(q))$.
After adding each $p_i$ to counter $i$, $M_c(q)$ guesses $x$ and simulates 
$M_1$ on the first $k_1$ counters from $q$ and simulates $M_2$ on the remaining counters and the pushdown from an initial configuration. It follows that $x \in L_2$,
and $(q, x\lhd, p_1, \ldots, p_{k_1}) \vdash_{M_1}^* (q_f^1, \lhd, p'_1, \ldots , p'_{k_1}), p_i' \geq 0, 1 \leq i \leq k_1, q_f^1 \in F_1$. Hence,
$w \in L_c(q)$.
\qed \end{proof}

Since for each $q \in Q_1$, $M_c(q)$ is in $\NPCM$, it accepts a semilinear language \cite{Ibarra1978},
and since the accepted language is bounded,
it is bounded-semilinear and can therefore be accepted by a $\DCM$-machine by
Theorem \ref{WWW}.
Let $M_c'(q)$ be this $\DCM$ machine, and let $k'$ be the maximum
number of counters of any $\DCM$ machine $M_c'(q),q \in Q_1$.

Thus, a final $\DCM$ machine $M'$ with $k_1+k'$ counters is built as follows. In it, $M'$ has $k_1$ counters used to simulate $M_1$, and also $k'$ additional counters, used to simulate some $M_c'(q)$, for some $q\in Q_1$. Then, $M'$ reads its input $x\lhd$, where $x\in \Sigma^*$,while simulating $M_1$ on the first $k_1$ counters, either failing, or 
reaching some configuration $(q, \lhd, p_1, \ldots, p_{k_1})$, for some $q\in Q_1$, upon first hitting the end-marker $\lhd$.
If it does not fail, we then simulate the $\DCM$-machine $M_c'(q)$ on input $a_1^{p_1} \cdots a_{k_1}^{p_{k_1}}$, but this simulation is done deterministically by subtracting $1$ from the first $k_1$ counters, in order, until each are zero instead of reading input characters, and accepts if $a_1^{p_1} \cdots a_{k_1}^{p_{k_1}} \in L(M_c'(q))= L_c(q)$. Then $M'$ is deterministic, and accepts
\begin{eqnarray*}
&& \{ x \mid \begin{array}[t]{l} \mbox{either~} (s_0, x\lhd, 0, \ldots, 0) \vdash_{M_1}^* (q', a\lhd, p_1', \ldots, p_{k_1}') \vdash_{M_1} (q, \lhd, p_1, \ldots, p_{k_1}), \\ a\in \Sigma, \mbox{~or~} (s_0, x\lhd, 0, \ldots, 0) = (q, \lhd, p_1, \ldots, p_{k_1}), \mbox{~s.t.~} a_1^{p_1} \cdots a_{k_1}^{p_{k_1}} \in L_c(q)\} \end{array} \\
& = & \{x \mid \begin{array}[t]{l} \mbox{either~} (s_0, x\lhd, 0, \ldots, 0) \vdash_{M_1}^* (q', a\lhd, p_1', \ldots, p_{k_1}') \vdash_{M_1} (q, \lhd, p_1, \ldots, p_{k_1}),  \\ a \in \Sigma, \mbox{~or~} (s_0, x\lhd, 0, \ldots, 0) = (q, \lhd, p_1, \ldots, p_{k_1}), \mbox{~where~}  \exists y \in L_2 \mbox{~s.t.~} \\
 (q,y\lhd, p_1, \ldots, p_{k_1}) \vdash_{M_1}^* (q_f, \lhd, p_1'', \ldots, p_{k_1}''), q_f \in F_1\}  \end{array} \\
& = & \{x \mid xy \in L_1, y \in L_2\} \\
& = & L_1 L_2^{-1}.
\end{eqnarray*}
\qed \end{proof}

This immediately shows closure for the prefix operation.

\begin{corollary}
\label{closedprefix}
If $L \in \DCM$, then $\pref(L) \in \DCM$.
\end{corollary}

We can modify this construction to show a strong closure result for
one-counter languages that does not increase the number of counters.
\begin{proposition}
\label{rightquotientwithNPCM}
Let $l \in \natnum$. If $L_1 \in \DCM(1,l)$ and $L_2 \in \NPCM$, then $\rquot{L_1}{L_2} \in \DCM(1,l)$. 
\end{proposition}
\begin{proof}
The construction is similar to the one in Proposition \ref{fullQuotientClosure}.
However, we note that since the input machine for $L_1$ has only one counter,
$L_c(q)$ is unary (regardless of the number of counters needed for $L_2$).
Thus $L_c(q)$ is unary and semilinear, and
Parikh's theorem states 
that all semilinear languages are letter-equivalent
to regular languages \cite{harrison1978}, and all unary semilinear languages are regular.
Thus $L_c(q)$ is regular, and can be accepted by a DFA.

We can then construct $M'$ accepting $\rquot{L_1}{ L_2}$ as in Proposition \ref{fullQuotientClosure}
without requiring any additional counters or counter reversals, by transitioning to the DFA
accepting $L_c(q)$ when we reach the end of input at state $q$.
\qed \end{proof}

\begin{corollary}
\label{closedprefixDCM1}
Let $l \in \natnum$. If $L \in \DCM(1,l)$, then $\pref(L) \in \DCM(1,l)$.
\end{corollary}

In fact, the constructions of Propositions \ref{fullQuotientClosure} and \ref{rightquotientwithNPCM} can be generalized from $\NPCM$ to any class of automata that can be defined using Definition \ref{counteraugmentable}. A condition to describe these classes of automata is described in more detail in \cite{Harju2002278}. But we only define the condition below in a way specific to our use in this paper. Only the first two conditions are required for Corollary \ref{generalizedSemilinear}, while the third is required for Corollary \ref{evenMoreGeneralSemilinear}.
\begin{definition}
\label{counteraugmentable}
A family of languages $\mathscr{F}$ is said to be {\em reversal-bounded counter augmentable} if
\begin{itemize}
\item every language in $\mathscr{F}$ is effectively semilinear,
\item given $\DCM$ machine $M_1$ with $k$ counters, state set $Q$ and final state set $F$, and $L_2 \in \mathscr{F}$, we
can effectively construct, for each $q\in Q$, 
the following language in $\mathscr{F}$,
$$\{ a_1^{p_1} \cdots a_k^{p_k} \sst \begin{array}[t]{l} \exists x \in L_2 \mbox{~such that ~} (q, x\lhd, p_1, \ldots, p_k) \vdash_{M_1}^* (q_f, \lhd, p'_1, \ldots p'_k), \\ p'_i \geq 0, 1 \leq i \leq k, q_f \in F \}, \end{array}$$

\item given $\DCM$ machine $M_1$ with $k$ counters, state set $Q$, initial state $q_0$, and $L_2 \in \mathscr{F}$, we
can effectively construct, for each $q\in Q$,
the following language in $\mathscr{F}$,
$$\{ a_1^{p_1} \cdots a_k^{p_k} \sst \exists x \in L_2 \mbox{~such that ~} (q_0, x, 0, \ldots,0) \vdash_{M_1}^* (q, \lambda, p_1, \ldots p_k)\}.$$
\end{itemize}
\end{definition}

\begin{corollary}
\label{generalizedSemilinear}
Let $L_1 \in \DCM$ and $L_2 \in \mathscr{F}$, a family of languages that is reversal-bounded counter augmentable.
Then $\rquot{L_1 }{ L_2} \in \DCM$. Furthermore, if $L_1 \in \DCM(1,l)$ for some $l \in \natnum$,
then $\rquot{L_1 }{ L_2} \in \DCM(1,l)$.
\end{corollary}


There are many reversal-bounded counter augmentable families that $L_2$ could be from in this corollary, such as:
\begin{itemize}
\item $\MPCA$'s: one-way machines with $k$ pushdowns where values may only be popped from the first non-empty stack,
augmented by a fixed number of reversal-bounded counters \cite{Harju2002278}. 
\item $\TCA$'s: nondeterministic Turing machines with a one-way read-only
input and a two-way read-write tape,
where the number of times the read-write head crosses any tape cell is finitely bounded,
again augmented by a fixed number of reversal-bounded counters \cite{Harju2002278}.
\item $\QCA$'s: $\NFA$'s augmented with a queue, where the number of alternations between the non-deletion
phase and the non-insertion phase is bounded by a constant \cite{Harju2002278},
augmented by a fixed number of reversal-bounded counters.
\item $\EPDA$'s: embedded pushdown automata, modelled around a stack of stacks, 
introduced in \cite{Vijayashanker:1987:STA:913947} augmented by a fixed
number of reversal-bounded counters. 
These accept the languages of tree-adjoining grammars,
a semilinear subset of the context-sensitive languages.
As was stated in \cite{Harju2002278}, we can augment this model with a fixed number of reversal-bounded counters
and still get an effectively semilinear family.
\end{itemize}

Finally, the construction of Proposition \ref{fullQuotientClosure} can
be used to show that deterministic one counter languages
(non-reversal-bounded) are closed under right quotient with $\NCM$.
\begin{proposition}
Let $L_1  \in \DCA$, and let
$L_2 \in \NCM$. Then $L_1 L_2^{-1} \in \DCA$.
\end{proposition}
\begin{proof}
Again, the construction is similar to Proposition \ref{fullQuotientClosure}.
However, since the input machine for $L_1$ has only one counter,
$L_c(q)$ is unary (regardless of the number of counters
needed for $L_2$). Then $L_c(q)$ is unary and is indeed an $\NPCM$
language, as $M_c(q)$ simulates $M_1$, this time using the
unrestricted pushdown to simulate the potentially non-reversal-bounded
counter of $M_1$, while simulating $M_2$ on the reversal-bounded counters.
Thus, because $\NPCM$ accept only semilinear languages \cite{Ibarra1978},
$L_c(q)$ is in fact a regular language and can be accepted by a DFA.
$M'$ can then be constructed to accept $L_1 L_2^{-1}$ without requiring any additional counters or counter reversals by transitioning to the DFA
accepting $L_c(q)$ when we reach the end of input at state $q$.
\qed \end{proof}

Next, for the 
case of one-counter machines that make only one counter reversal, it will be shown that a $\DCM$-machine that can accept their suffix and infix languages can always be constructed. However, in some cases, these resulting machines often require more than one counter.
Thus, unlike prefix, $\DCM(1,1)$ is not closed under suffix, left quotient, or infix. But, the result is in $\DCM$.

As the proof is quite lengthy, we 
will give some intuition for the result first. First, $\DCM$ is closed under union \cite{Ibarra1978} (following from closure under intersection and complement) and so the second statement of Proposition \ref{leftquotientwithNPCM} follows from the first. For the first statement, an intermediate $\NPCM$ machine is constructed from $L_1$ and $L$ that accepts a language $L^c$ (here, $c$ is a superscript label rather than an exponent). This language contains words of the form
$qa^i$ where there exists some word $w$ such that both $w \in L_1$, and also 
from the initial configuration of $M$ (accepting $L$), it can read $w$ and reach state $q$ with $i$ on the counter. Then, it is shown that this language is actually a regular language, using the fact that all semilinear unary languages are regular. Then, $\DCM(1,1)$ machines are created for every state $q$ of $M$. These accept all words $w$ such that 
$qa^i \in L^c$, and in $M$, from state $q$ and counter $i$ with $w$ to read as input, $M$ can reach a final state while emptying the counter.
The fact that $L^c$ is regular allows these machines to be created.

\begin{proposition}
\label{leftquotientwithNPCM}
Let $L \in \DCM(1,1), L_1 \in \NPCM$. Then 
$L_1^{-1} L$ is the finite union of languages in $\DCM(1,1)$. Furthermore, it is in $\DCM$.
\end{proposition}
\begin{proof}
For the first statement, let $M_1$ be an $\NPCM$ machine accepting
$L_1$, and 
let $M = (1,Q,\Sigma, \lhd, \delta, q_0,F)$ be a 1-reversal bounded, 1-counter machine accepting $L$. 

Next, we will argue that it is possible to assume, without loss of generality, that $M$ has the following form:
\begin{enumerate}
\item $Q = Q\down \cup Q\up$,
\item for all $q \in Q\down$, all transitions defined on $q$ either decrease the counter or keep it the same,
\item for all $q \in Q\up$, all transitions defined on $q$ either increase the counter or keep it the same,
\item the sequence of states $p_0 p_1 \cdots p_n, n \geq 0, p_0 = q_0, n \geq 0$ traversed in every computation from the 
initial configuration satisfies $p_0 \cdots p_i \in Q\up^*, p_{i+1} \cdots p_n \in Q\down^*, 0 \leq i \leq n$, 
with the transition from $p_{i+1}$ to $p_{i+2}$ being the (first, if it exists) decreasing transition,
\item for all states $q \in Q\down$ all stay transitions
defined on $q$ (except on $\delta(q,\lhd,0)$) change the counter,
\item $\delta(q,d,1)$ is defined for all $q \in Q, d\in \Sigma$, 
\item the counter always empties before accepting.
\end{enumerate}
Indeed, it is possible to transform a $\DCM(1,1)$ machine of the form of $M$ to another $\DCM(1,1)$ machine $\bar{M} = (1,\bar{Q}, \Sigma, \lhd, \bar{\delta}, q_0, \bar{F})$ satisfying the first four
conditions, as follows. First, 
let $Q\up = Q$, and $Q\down =  \{q' \mid q \in Q\}$ (primed versions of the state set), $\bar{Q} = Q\up \cup Q\down$, 
and $\bar{F} = F \cup \{q_f' \mid q_f \in F\}$. 
Then, for all transitions that either decrease the counter
or keep it the same, $(p,T, j) \in \delta(q,c,i), i \in \{0,1\}, j \in \{0, -1\}$, 
instead create the transition $(p',T, j) \in \bar{\delta}(q',c,i)$.
Further, for all transitions of $\delta$ that either increase the counter or keep it the same, keep this transition in $\bar{\delta}$.
Then, the first three conditions are satisfied.
Then, for all those decreasing transitions $(p,T, -1) \in \delta(q,c,1)$,
add in $(q', {\rm S}, 0) \in \bar{\delta}(q,c,1)$. Therefore, condition four is satisfied.
It is clear that $L(\bar{M}) = L(M)$, as any sequence of transitions with at most one counter reversal (all accepting
computations have at most one reversal) can traverse the same transitions (using states in $Q\up$) 
until the first decrease transition, at which point
only the new stay transitions from $Q\up$ to $Q\down$ are defined, and then the computation continues with transitions on $Q\down$.
This machine can be further transformed into one accepting the same language and additionally satisfying conditions 5, 6, and 7 as
any sequence of stay transitions that does not change the counter can be ``skipped over'' to either a right transition or a decrease transition,
a ``dead state'' can be added to satisfy condition 6, and the states can enforce 7.

Therefore, assume without loss of generality that $M$ satisfies these conditions. This will simplify the rest of the construction.

Next, we create an $\NPCM$ machine $M'$ that accepts 
$$L^c = \{qa^i \mid \exists w \in L_1, (q_0, w, 0) \vdash_M^* (q,\lambda, i)\},$$
where $a$ is a new symbol not in $\Sigma$. Indeed, $M'$ operates by nondeterministically guessing a word $w\in \Sigma^*$, letter-by-letter, and simulating in parallel (using stay transitions), the $\NPCM$ machine $M_1$ using the pushdown and a set of counters, as well as simulating $M$ on $w$ on an additional counter. Then, after reading the last
letter of the guessed $w$, $M'$ verifies that the simulated machine $M$ is in state $q$ (reading the state $q$ as part of the input), and verifies that the contents of the simulated counter of $M$ is $i$, matching the input. Then, it verifies that $w$ is in $L_1$ by continuing the
simulation of $M_1$ on the end-marker. Furthermore, for each $q \in Q$, the set $q^{-1}L^c$ is a unary $\NPCM$ language (as discussed in Section
\ref{sec:prelims}, $\NPCM$ is closed under left quotient with regular languages). Indeed, every $\NPCM$ language is semilinear \cite{Ibarra1978}, and it is also known that every unary semilinear language  is regular \cite{harrison1978}, and effectively constructable. Thus, $L^c = \bigcup_{q\in Q}(q(q^{-1}L^c))$ is regular as well. Let $M^c = (Q^c, Q \cup\{a\}, \delta^c, s_0^c, F^c)$ be a DFA accepting $L^c$. Assume without loss of generality that $M^c$ is a complete DFA.

For the remainder of the proof, the layout will proceed by
creating three sets of $\DCM(1,1)$ machines and languages as follows:
\begin{enumerate}
\item $M_0^q$, for all $q\in Q$, and $L_0^q = L(M_0^q)$. We will construct it such that
\begin{equation}
L_0^q= \{w \mid (q,w\lhd,0) \vdash_M^* (q_f, \lhd, 0), q_f \in F, qa^0 = q \in L^c\}.
\label{bigequation1}
\end{equation}

\item $M^q\up$, for all $q\in Q\up$, and $L^q\up = L(M^q\up)$. We will construct it such that
\begin{equation}
L^q\up= \{w \mid \exists i >0, (q,w\lhd,i) \vdash_M^* (q_f, \lhd, 0), q_f \in F, qa^i \in L^c\}.
\label{bigequation2}
\end{equation}

\item $M^q\down$, for all $q\in Q\down$, and $L^q\down = L(M^q\down)$. We will construct it such that
\begin{equation}
L^q\down= \{w \mid \exists i > 0, (q,w\lhd,i) \vdash_M^* (q_f, \lhd, 0), q_f \in F, qa^i \in L^c\}.
\label{bigequation3}
\end{equation}
\end{enumerate}
It is clear that $$L_1^{-1}L(M) = \bigcup_{q\in Q}L_0^q \cup  \bigcup_{q\in Q\up}L^q\up \cup  \bigcup_{q\in Q\down}L^q\down,$$ 
and thus it suffices to build the $\DCM(1,1)$ machines and show that Equations (\ref{bigequation1}), (\ref{bigequation2}), and (\ref{bigequation3}) hold.
We will do this with each type next.

First, for (\ref{bigequation1}), construct $M_0^q$ for $q\in Q$ as follows: $M_0^q$ operates just like $M$ starting at state $q$ if $q \in L^c$, and if $q \notin L^c$, then it accepts $\emptyset$. Hence, (\ref{bigequation1}) is true.

Next, we will show (\ref{bigequation3}) is true. 
It will be shown that $L^q\down$ is a regular language. Then the construction and proof of correctness of (\ref{bigequation3}) will be used within the proof and construction of (\ref{bigequation2}). A slight generalization of (\ref{bigequation3}) will be used in order to accommodate its use for (\ref{bigequation2}). Despite the languages being regular, $\NCM(1,1)$ machines will be built to accept them, but without
using the counter (for consistency of notation, and to use nondeterminism). It is immediate that these are regular, can be converted to NFAs, then to DFAs, then to $\DCM(1,1)$ machines that do not use the counter.
Intuitively, each $\NCM(1,1)$ machine (for each $q \in Q\down$) will simulate $M$ starting at state $q$, but then only non-increasing transitions can be used, as only transitions on $Q\down$ can be reached from $q$. However, instead of decreasing from a counter,
the $\NCM(1,1)$ machine instead simulates the DFA $M^c$ in parallel, and reads a single $a$ for every decrease of the simulated computation of $M$.
If the simulated computation of $M^c$ is in a final state, then the counter could be zero in the simulated computation of $M$ and reach that configuration. But the simulated computation of $M$ may only accept from configurations with larger counter values (depending on the remaining sequence of transitions). Thus, the new machine uses nondeterminism to try every possible configuration where zero could occur on the counter, trying each to see if the rest of the input accepts (by directly simulating $M$).

We will give the construction here of the intermediate $\NCM(1,1)$ machines that do not use the counter, then the proof of correctness of the construction. All the machines $\overline{M^{q,q'}\down}\in \NCM(1,1)$, for each $q \in Q\down, q' \in Q$ will have the same set of input alphabets, states, transitions, and final states, with only the initial state differing.

Formally, let $q \in Q\down, q' \in Q, q_0^c = \hat{\delta}^c(s_0^c,q')$. Then $\overline{M^{q,q'}\down} = (1,P\down, \Sigma,\lhd,\delta\down, s^{q,q'}\down,F\down)$, where $P\down = (Q\times Q^c) \cup Q\down, s^{q,q'}\down = (q,q_0^c), F\down = F$.

The transitions of $\delta\down$ are created (none using the counter) by the following algorithm:
\begin{enumerate}

\item\label{itm:subtypeA} For all transitions $(p,{\rm S},-1) \in \delta(r,d,1), p,r \in Q\down, d \in \Sigma \cup \{\lhd\}$, and all $r^c \in Q^c$, create
$$((p,\delta^c(r^c,a)),{\rm S}, 0) \in \delta\down((r,r^c),d,0),$$
and if $\delta^c(r^c,a) \in F^c$, create
$$(p,{\rm S},0) \in \delta\down((r,r^c),d,0).$$
\item\label{itm:subtypeB} For all transitions $(p,{\rm R},0) \in \delta(r,d,1), p,r \in Q\down, d \in \Sigma$, and all $r^c \in Q^c$, create
$$((p,r^c),{\rm R},0) \in \delta\down((r,r^c),d,0).$$
\item\label{itm:subtypeC} For all transitions $(p,{\rm R}, -1) \in \delta(r,d,1), p,r \in Q\down, d \in \Sigma$, and all $r^c \in Q^c$, create
$$((p,\delta^c(r^c,a)),{\rm R},0) \in \delta\down((r,r^c),d,0),$$
and if $\delta^c(r^c,a) \in F^c$, create
$$(p,{\rm R},0) \in \delta\down((r,r^c),d,0).$$

\item For all transitions $(p,{\rm R},0) \in \delta(r,d,0), p,r \in Q\down, d \in \Sigma$, create $$(p,R,0) \in \delta\down(r,d,0).$$ 

\item For all transitions $(p,{\rm S},0) \in \delta(r,\lhd,0), p,r \in Q\down$, create $$(p,{\rm S},0) \in \delta\down(r,\lhd,0).$$
\end{enumerate}
The states of the machine consist of ordered pairs in $Q \times Q^c$ in addition to states of $Q\down$ to allow the simulation of a
$M$ and $M^c$ in parallel, until the computation of $M^c$ can hit a final state, at which point it can (optionally) switch to only simulating $M$
on an empty counter.
Transitions created in step 1 simulate all decreasing transitions that stay on the input, by reading an $a$ from the simulation computation 
of $M^c$ to simulate the decrease; and if the simulated computation of $M^c$ can reach a final state, the simulated computation
of $M^c$ can optionally end.
Transitions created in step 2 simulate all transitions of $M$ defined on a positive counter that move right on the input but
do not change the counter (thus not changing the state of $Q^c$). Transitions created in step 3 simulate all transitions of
$M$ that move right on the input and decrease the counter by $M^c$ reading $a$ and optionally ending. Transitions created in step 4 and 5
simulate those of $M$ defined on an empty counter verbatim.

Intuitively, the next claim demonstrates that it is possible to simulate $M$ starting at $q$ and counter value $i$ whereby $q' a^i \in L^c$
as follows: first it simulates the computation of $M$ starting at $q$ using the first component, and each decrease in counter reads $a$ from the simulated computation of $M^c$. This continues until the counter reaches zero after $i$ decreasing transitions, at which point the simulated computation of $M^c$ is
in a final state. Then, the simulation of $M$ can continue verbatim. The formal proof is presented next.
\begin{claim}
\label{claim:subsetlang}
For all $q\in Q\down, q' \in Q$, $$\{w \mid \exists i > 0, (q,w \lhd,i) \vdash_M^* (q_f, \lhd, 0), q_f \in F, q'a^i \in L^c\} \subseteq L(\overline{M^{q,q'}\down}).$$
\end{claim}
\begin{proof}
Let $q \in Q\down, q' \in Q$. 
Let $w$ be such that there exists $i> 0, q_f\in F, q'a^i \in L^c$, and $(q,w \lhd,i) \vdash_M^* (q_f, \lhd,0)$. Let $p_j,w_j,x_j, 0 \leq j \leq m$ be such that $p_0=q, w=w_0, x_0=i, q_f = p_m,  w_m =\lambda,x_m=0$ and $(p_l,w_l \lhd,x_l)\vdash_M (p_{l+1},w_{l+1} \lhd ,x_{l+1}), 0 \leq l <m$, via transition $t_{l+1}$. Then
$$(p_0,w_0 \lhd ,x_0)\vdash_M^* (p_{\gamma},w_{\gamma} \lhd,x_{\gamma}) \vdash_M^* (p_m, w_m \lhd,x_m),$$
where $\gamma$ is the smallest number such that $x_{\gamma} <i$ (it exists since $i>0$), and $\mu$ the smallest number greater than or equal to $\gamma$ such that $x_{\mu} = 0$.

The transitions $t_1, \ldots, t_{\gamma-1}$ are of the form, for $0 \leq l <\gamma -1$,
$(p_{l+1}, T_{l+1},y_{l+1}) \in \delta(p_l,d_l,1)$, where $i$ is on the counter on all $x_0, \ldots, x_{\gamma-1}$ (since $x_0 = i$, and $x_{\gamma}$ is the first counter value less than $\gamma$), and $y_0, \ldots, y_{\gamma-1}$ are all equal to $0$. These must all be right transitions since they do not change the counter and so they create transitions in step \ref{itm:subtypeB} of the construction, of the form 
$$((p_{l+1},q_0^c),  R,0) \in \delta\down((p_l, q_0^c), d_l, 0),$$ for $0 \leq l < \gamma -1$.
Then, $$((p_0, q_0^c),w_0 \lhd, x_0 - i = 0) \vdash_{\overline{M^{q,q'}\down}}^* ((p_{\gamma-1}, q_0^c), w_{\gamma-1} \lhd, x_{\gamma-1}-i=0).$$

The transitions $t_{\gamma}, \ldots, t_{\mu}$ are of the form, for $\gamma -1 \leq l < \mu$, $(p_{l+1}, T_{l+1}, y_{l+1}) \in \delta(p_l, d_l, 1)$, and for $\gamma - 1 \leq l  < \mu -1$ ($t_{\mu}$ is the last decreasing transition), creates transitions in steps \ref{itm:subtypeA}, \ref{itm:subtypeB}, and \ref{itm:subtypeC} of the form $$((p_{l+1}, q_{l+1}^c),  T_{l+1},0) \in \delta\down((p_l, q_l^c), d_l, 0),$$
for some $q_l^c, q_{l+1}^c \in Q^c$.

Then,
$$((p_{\gamma-1},q_0^c),w_{\gamma -1} \lhd, 0) \vdash_{\overline{M^{q,q'}\down}} \cdots \vdash_{\overline{M^{q,q'}\down}} ((p_{\mu-1}, q_{\mu-1}^c),w_{\mu-1} \lhd,0),$$
where there are exactly $i-1$ decreasing transitions being
simulated in this sequence. From $q_{\mu -1}^c$, reading one more $a$,
$\delta^c(q_{\mu-1},a) \in F^c$ since $q'a^i \in F^c$, and thus $(p_{\mu},T_{\mu},y_{\mu}) \in \delta(p_{\mu-1},d_{\mu-1},1)$ creates 
$(p_{\mu},T_{\mu},y_{\mu}) \in \delta\down((p_{\mu-1},q_{\mu-1}^c),d_{\mu-1},0)$ in step \ref{itm:subtypeA} or \ref{itm:subtypeC}.

Then there remains transitions $t_{\mu+1}, \ldots, t_{m}$, for $\mu \leq l < m$ of the form $(p_{l+1},  T_{l+1},0) \in \delta(p_l, d_l,0)$.
These transitions are all in $\delta\down$ and thus
$$(p_{\mu}, w_{\mu}\lhd ,0) \vdash^*_{\overline{M^{q,q'}\down}} (p_m = q_f, \lhd,0),$$
and hence $w\in L(\overline{M^{q,q'}\down})$.

\qed \end{proof}

The converse can be seen by examining an arbitrary computation of $\overline{M^{q,q'}\down}$ which must have two components in the states, corresponding
to a simulation of $M$ in the first component and $M^c$ in the second component, until some
configuration where it switches to one component, continuing the simulation of $M$. The number of transitions used that simulate the
reading of an $a$ from the second component must be some $i$, where $q' a^i \in L^c$, and therefore a computation of $M$ can proceed
as in the simulation starting with $i$ in the counter, and reach a final state. The formal proof is presented next.
\begin{claim}
\label{claim:supersetlang}
For all $q\in Q\down, q' \in Q$, $$L(\overline{M^{q,q'}\down}) \subseteq \{w \mid \exists i > 0, (q,w \lhd,i) \vdash_M^* (q_f, \lhd, 0), q_f \in F, q'a^i \in L^c\}.$$
\end{claim}
\begin{proof}
Let $w \in L(\overline{M^{q,q'}\down}), q \in Q\down, q' \in Q$. Let $\mu$ ($\mu$ is the last position
of the derivation with an ordered pair as state), $p_l, w_l, 0 \leq l \leq m$, and $q_j^c, 0 \leq j \leq \mu < m$ be such that $p_0 = q, w_0=w, w_m = \lambda, q_m \in F,$ and 
$$((p_l,q_l^c), w_l \lhd, 0) \vdash_{\overline{M^{q,q'}\down}} ((p_{l+1}, q_{l+1}^c), w_{l+1} \lhd, 0),$$ 
for $0 \leq l < \mu$, via transition $t_{l+1}$ of the form $((p_{l+1},q_{l+1}^c),T_{l+1},0) \in \delta\down((p_l,q_l^c),d_l,0)$, and 
$$((p_{\mu},q_{\mu}^c),w_{\mu} \lhd ,0) \vdash_{\overline{M^{q,q'}\down}} (p_{\mu+1},w_{\mu+1} \lhd, 0),$$ 
via transition $t_{\mu+1}$ of the form $(p_{\mu+1},T_{\mu+1},0) \in \delta\down((p_{\mu},q_{\mu}^c),d_{\mu},0)$ and $$(p_l,w_l \lhd,0) \vdash_{\overline{M^{q,q'}\down}} (p_{l+1},w_{l+1} \lhd, 0),$$
for $\mu+1 \leq l < m$ via transitions $t_{l+1}$ of the form $(p_{l+1}, T_{l+1},0 ) \in \delta\down(p_l,d_l,0)$.
Let $i$ be the number of times transitions created in step \ref{itm:subtypeA} or \ref{itm:subtypeC} are applied. Then by the transition $t_{\mu +1}$, this implies $q'a^i \in F^c$. Then, this implies that there are transitions 
$(p_{l+1}, T_{l+1},y_{l+1}) \in \delta(p_l,d_l,1)$, for all $l, 0 \leq l \leq \mu$, with $i$ decreasing transitions and
$(p_{l+1}, T_{l+1},0) \in \delta(p_l,d_l,0)$, for all $l, \mu+ 1 \leq l < m$, by the construction. Hence, the claim follows.
\qed \end{proof}

We let $M^{q,q'}= (1,Q^{q,q'},\Sigma,\lhd, \delta^{q,q'}\down,s^{q,q'}\down, F^{q,q'}\down)$ be a $\DCM(1,1)$ machine (that is hence deterministic) accepting $L(\overline{M^{q,q'}})$ that never uses the counter, which can be created since it is regular. Assume all the sets of states of different machines $Q^{q,q'}\down$ are disjoint.

Then, to prove Equation (\ref{bigequation3}), only machines $M^{q}\down = M^{q,q}\down, q \in Q\down$ accepting the languages 
$L^{q,q}\down, q \in Q\down$ need to be considered, and they are all indeed regular.

The construction for $M^q\up$ will be given next, and it will use the transitions from the machines $M^{r,q}\down$ within it.
Intuitively, $M^q\up$ will simulate computations of $M$ that start from configuration $(q,u\lhd, i)$ up to a maximum counter value of
$\alpha$ and back to counter value $i$ again. However, these computations are simulated by 
starting at a counter value of $0$ instead of $i$ (from $(q,u\lhd,0)$) to a maximum of $\alpha-i$ (instead of $\alpha$),
back to $0$ again (instead of $i$), ending at a configuration of the form $(r,u'\lhd,0)$. Thus, the simulated computation takes place with $i$ subtracted from each counter value
of each configuration.
Then, $M^q\up$ uses the machine $M^{r,q}\down$ to test if the rest of the input can be accepted starting at $r$ with any counter value that can reach $q$ by using words in $L^c$ that start with $q$.

Formally, for $q \in Q\up$,
$M^q\up = (1,P\up, \Sigma,\lhd,\delta\up, s^q\up,F\up)$, where $P\up = Q\cup \bigcup_{r \in Q} Q^{r,q}\down, s^q\up = q, F\up = \bigcup_{r \in Q\down}F^{r,q}$, where $Q$ is disjoint from other states.

The transitions of $\delta\up$ are created by the following algorithm:
\begin{enumerate}
\item\label{itm:fromoriginal} For all transitions $(p,T,y) \in \delta(r,d,1), p,r \in Q, d\in \Sigma \cup \{\lhd\}, T \in \{{\rm S}, {\rm R}\}, y \in \{-1,0,1\}$, create
$$(p,T,y) \in \delta\up(r,d,e),$$
for both $e=1$, and $e=0$ if $r \in Q\up$,
\item\label{itm:switchtobigstate} Create $(s^{r,q}\down,{\rm S},0) \in \delta\up(r,d,0)$, for all $ d\in \Sigma \cup \{\lhd\}$, and for all $ r \in Q\down$,
\item\label{itm:fromdown} Add all transitions from $M^{s,q}\down, s \in Q\down$.

\end{enumerate}

Indeed, $M^q\up$ is deterministic as those transitions created in step \ref{itm:fromoriginal} are in $M$, and $M^{s,p}\down$ is deterministic, for all $s,p$.

\begin{claim}
For all $q\in Q\up$, $$\{w \mid \exists i >0, (q,w \lhd,i) \vdash_M^* (q_f, \lhd, 0), q_f \in F, qa^i \in L^c\} \subseteq L^q\up.$$
\end{claim}
\begin{proof}
Let $q \in Q\up$. 
Let $w$ be such that there exists $i>0, q_f\in F, qa^i \in L^c$, and $(q,w\lhd,i) \vdash_M^* (q_f, \lhd,0)$. Let $p_j,w_j,x_j, 0 \leq j \leq m$ be such that $p_0=q, w=w_0, x_0=i, q_f = p_m, \lambda = w_m,x_m=0$ and $(p_l,w_l\lhd,x_l)\vdash_M (p_{l+1},w_{l+1}\lhd,x_{l+1}), 0 \leq l <m$, via transition $t_{l+1}$. Assume that there exists $\alpha >1$ such that $x_{\alpha} > i$, and let $\alpha$ be the smallest such number. Then, there exists
$$(p_0,w_0\lhd,x_0)\vdash_M^* (p_{\alpha},w_{\alpha}\lhd,x_{\alpha}) \vdash_M^* (p_{\beta},w_{\beta}\lhd,x_{\beta})\vdash_M^* (p_m, w_m\lhd,x_m),$$
where $\beta$ is the smallest number bigger than $\alpha$ such that $x_{\beta} = i$. In this case, in step \ref{itm:fromoriginal} of the algorithm, transitions $t_1, \ldots, t_{\alpha}$ of the form $(p_l,T_l,y_l) \in \delta(p_{l-1}, d_{l-1},1), 0 < l \leq \alpha$, create transitions of the form $(p_l,T_l,y_l) \in \delta\up(p_{l-1}, d_{l-1},0)$, and thus 
$$ (p_0,w_0\lhd,x_0-i=0) \vdash_{M^q\up}^* (p_{\alpha-1},w_{\alpha-1}\lhd, x_{\alpha-1} - i = 0)\vdash_{M^q\up} (p_{\alpha},w_{\alpha}\lhd,x_{\alpha}-i),$$ where $x_{\alpha} - i > 0$.

In step \ref{itm:fromoriginal} of the algorithm, transitions $t_{\alpha+1}, \ldots, t_{\beta}$ of the form $(p_l,T_l,y_l ) \in \delta(p_{l-1},d_{l-1},1), \alpha < l \leq \beta$ create transitions of the form $$(p_l,T_l, y_l) \in \delta\up(p_{l-1}, d_{l-1}, 1).$$ Thus,
$(p_{\alpha},w_{\alpha}\lhd, x_{\alpha}-i) \vdash_{M^q\up}^* (p_{\beta}, w_{\beta}\lhd, x_{\beta}-i = 0)$, since $x_{\alpha}-i, \ldots, x_{\beta-1}-i$ are all greater than $0$.
Then, using a transition of type \ref{itm:switchtobigstate}, $(p_{\beta}, w_{\beta}\lhd,0)\vdash_{M^q\up} (s^{p_{\beta},q}\down,w_{\beta}\lhd,0)$. Then since
$(p_{\beta},w_{\beta}\lhd, x_{\beta}) \vdash_M^* (p_m,\lhd,0), p_m \in F$, and $p_{\beta} \in Q\down, qa^i \in L^c$, then $w_{\beta} \in L^{p_{\beta},q}\down$, by Claim \ref{claim:subsetlang}. Hence,
$$(s^{p_{\beta},q}\down,w_{\beta}\lhd,0) \vdash^*_{M^{p_{\beta},q}\down} (q_f',\lhd,0),$$
$q_f' \in F$, and therefore, this occurs in $M^q\up$ as well.

Lastly, the case where there does not exist an $\alpha >i$ such that $x_{\alpha}>i$ (thus $i$ is the highest value in counter) is similar, by applying transitions of type \ref{itm:fromoriginal} until the transitions before the first decrease (the first time a state from $Q\down$
is reached), then a transitions of type \ref{itm:switchtobigstate}, followed by a sequence of type \ref{itm:fromdown} transitions as above.

\qed \end{proof}

The reverse containment can be shown by examining any accepting sequence of configurations, which has some initial simulation of $M$, 
followed by a computation of a machine $M\down^{q',q}$. The initial simulation can occur in $M$ with $i >0 $ added to each counter
value, and the correctness of the remaining portion of $M\down^{q',q}$ follows from Claim \ref{claim:supersetlang}.

\begin{claim}
For all $q\in Q\up$, $$L^q\up \subseteq \{w \mid \exists i >0, (q,w\lhd,i) \vdash_M^* (q_f, \lhd, 0), q_f \in F, qa^i \in L^c\}.$$
\end{claim}
\begin{proof}
Let $w \in L(M^q\up)$. 
Then $$(q,w\lhd,0)\vdash_{M^q\up}^* (q', w'\lhd,0) \vdash_{M^q\up} ((q',\delta^c(s_0^c,q)),w'\lhd,0) \vdash_{M^q\up}^* (q_f',\lhd,0),$$ where $q_f' \in F^{q',q}$.
Let $\beta, p_l, w_l, x_l, 0 \leq l \leq \beta$ be such that $p_0=q, w_0=w, x_0=0, q' = p_{\beta}, w' = w_{\beta}, x_{\beta}=0$
such that $(p_l,w_l\lhd,x_l) \vdash_{M^q\up} (p_{l+1},w_{l+1}\lhd, x_{l+1}), 0 \leq l < \beta$.

Then $w' \in L^{q',q}\down$, and therefore by Claim \ref{claim:supersetlang}, there exists $i > 0$ such that $(q',w'\lhd,i) \vdash_M^* (q_f, \lhd, 0), q_f \in F, qa^i \in L^c$. By the construction in step \ref{itm:fromoriginal},
$$(p_0,w_0\lhd, x_0+i) \vdash_M \cdots \vdash_M (p_{\beta},w_{\beta}\lhd,x_{\beta}+i),$$
and since $x_0 = x_{\beta}=0$ and $w' = w_{\beta}$ and $q' = p_{\beta}$, then
$(q,w\lhd,i) \vdash_M^* (q_f,\lhd,0)$ and $qa^i \in L^c$ and the claim follows.

\qed \end{proof}
Hence, Equation \ref{bigequation2} holds.

It is also known that $\DCM$ is closed under union (by increasing the number of counters) \cite{Ibarra1978}. Therefore, the finite union is in $\DCM$.

\qed \end{proof}

From this, we obtain the following general result.
\begin{proposition}
Let $L \in \DCM(1,1), L_1, L_2 \in \NPCM$. Then both
$(L_1^{-1}L)L_2^{-1}$ and $L_1^{-1}(L L_2^{-1})$ are a finite union of languages in $\DCM(1,1)$. Furthermore, both languages are in $\DCM$.
\end{proposition}
\begin{proof}
It will first be shown that $(L_1^{-1}L)L_2^{-1}$ is the finite union of languages in $\DCM(1,1)$.
Indeed, $L_1^{-1}L$ is the finite union of languages in
$\DCM(1,1), 1 \leq i \leq k$ by Proposition \ref{leftquotientwithNPCM}, and so $L_1^{-1}L = \bigcup_{i=1}^k X_i$ for $X_i \in \DCM(1,1)$. Further, for each $i$, $X_i L_2^{-1}$ is the finite union of $\DCM(1,1)$ languages by Proposition \ref{rightquotientwithNPCM}.

It remains to show that $\bigcup_{i=1}^k X_i L_2^{-1} =  (L_1^{-1}L)L_2^{-1}$. 
If $w\in \bigcup_{i=1}^k X_i L_2^{-1}$, then $w \in X_i L_2^{-1}$ for some $i$, $1 \leq i \leq k$, then $wy \in X_i, y \in  L_2$. Then $wy \in L_1^{-1}L$, and $w \in (L_1^{-1} L) L_2^{-1}$.
Conversely, if $w \in (L_1^{-1}L)L_2^{-1}$, then $wy \in L_1^{-1}L$ for some $y \in L_2$, and so $wy \in X_i$ for some $i$, $1 \leq i \leq k$, and thus $w \in X_i L_2^{-1}$.

For $L_1^{-1}(L L_2^{-1})$,
it is true that $L L_2^{-1} \in \DCM(1,1)$ by Proposition \ref{rightquotientwithNPCM}. Then $L_1^{-1}(L L_2^{-1})$ is the finite union of $\DCM(1,1)$ by Proposition \ref{leftquotientwithNPCM}.

It is also known that $\DCM$ is closed under union (by increasing the number of counters) \cite{Ibarra1978}. Therefore, both finite unions are in $\DCM$.
\qed \end{proof}
And, as with Corollary \ref{generalizedSemilinear}, this can be generalized to any language families that are reversal-bounded counter augmentable.

\begin{corollary}
\label{evenMoreGeneralSemilinear}
Let $L \in \DCM(1,1), L_1 \in \mathscr{F}_1, L_2 \in \mathscr{F}_2$, where $\mathscr{F}_1$ and $\mathscr{F}_2$ are any families of languages that are reversal-bounded counter augmentable.
Then 
$(L_1^{-1}L)L_2^{-1}$ and $L_1^{-1}(L L_2^{-1})$ are both a finite union of languages in $\DCM(1,1)$. Furthermore, both languages are in $\DCM$.
\end{corollary}

As a special case, when using the fixed regular language $\Sigma^*$ for the right and left quotient, we obtain:
\begin{corollary}
\label{suffinfDCM}
Let $L \in \DCM(1,1)$. Then 
$\suff(L)$ and $\infx(L)$ are both $\DCM$ languages. 
\end{corollary}

It is however sometimes necessary that the number of counters increase to accept $\suff(L)$ and $\infx(L)$, when $L \in \DCM(1,1)$ as seen from the next Proposition indicating that the suffix, infix, and outfix of a $\DCM(1,1)$ language can be outside of $\DCM(1,1)$.
\begin{proposition}
\label{suff11}
There exists $L \in \DCM(1,1)$ where all of $\suff(L), \infx(L), \outf(L)$ are not in $\DCM(1,1)$.
\end{proposition}
\begin{proof}
Assume otherwise. Let $L = \{a^n b^n c^n \mid n \geq 0\}, L_1 = \{  a^n b^n c^k  \mid n,k \geq 0\}, L_2=\{a^n b^m c^m \mid n,m \geq 0\}, L_3 = \{a^n b^m c^k  \mid n,m,k \geq 0\}$. Let $\Sigma = \{a,b,c\}$ and $\Gamma = \{d,e,f\}$.

It is well-known that $L$ is not a context-free language, and therefore is not a $\DCM(1,1)$ language. However, each of $L_1, L_2, L_3$ are $\DCM(1,1)$ languages, and therefore, so are $\overline{L_1}, \overline{L_2}, \overline{L_3}$ \cite{Ibarra1978} and so is $L' = d \#_1\overline{L_1} \#_2 \cup e \#_1\overline{L_2} \#_2 \cup f \#_1 \overline{L_3} \#_2$ (all complements with respect to $\Sigma^*$). 
The symbols $d,e,f$ are needed here, as each deterministically 
triggers the computation of a different $\DCM(1,1)$ machine so that the resulting machine can
be a $\DCM(1,1)$ machine (although $\DCM$ is closed under union, this closure can increase the number of counters; but this type of marked
union does not increase the number of counters).
It can also be seen that $\overline{L} = \overline{L_1} \cup \overline{L_2} \cup \overline{L_3}$.

But $\suff(L') \cap \#_1\Sigma^* \#_2 =  \infx(L') \cap \#_1\Sigma^* \#_2 = \outf(L') \cap \#_1\Sigma^* \#_2 = \#_1 \overline{L} \#_2$, and since $\DCM(1,1)$ is closed under intersection with regular languages, under left and right quotient by a symbol, and under complement, this implies $L$ is a $\DCM(1,1)$ language, a contradiction.
\qed \end{proof}

\section{Non-Closure Under Suffix, Infix, and Outfix for Multi-Counter and Multi-Reversal Machines}
\label{sec:nonclosure}

In \cite{EIMInsertion2015}, a technique was used to show that languages are not in $\DCM \cup 2\DCM(1)$. 
The technique uses undecidable properties to show non-closure. As $2\DCM(1)$ machines have a two-way input and a reversal-bounded counter, it is difficult to derive ``pumping'' lemmas for these languages. Furthermore, unlike $\DCM$ and $\NCM$ machines, $2\DCM(1)$ machines can accept non-semilinear languages.  For example,
$L_1= \{a^i b^k ~|~ i, k \ge 2, i$ divides $k \}$ can
be accepted by a 2\DCM(1) whose counter makes only one reversal.
However, $L_2 = \{a^i b^j c^k ~|~ i,j,k \ge 2,
k = ij \}$ cannot be accepted by a 2\DCM(1) \cite{IbarraJiang}.
This technique from \cite{EIMInsertion2015} works as follows.
The proof uses the fact that there is a recursively enumerable
but not recursive
language $L_{\rm re} \subseteq \natzero$ that is accepted by a deterministic
2-counter machine \cite{Minsky}. 
Here, these machines do not have an input tape, and acceptance is defined
whereby $n \in \natzero$ is accepted (i.e., $n \in L_{\rm re}$) if and only if, 
when started with $n$ in the first counter (encoded in unary)
and zero in the second counter, 
$M$ eventually halts (hence, acceptance is by halting).

Examining the constructions  in \cite{Minsky}
of the 2-counter machine demonstrates
that the counters behave in a regular pattern. 
Initially one counter has some
value $d_1$ and the other counter is zero. Then, 
the machine's operation can be divided into phases, where each
phase starts with one of the counters equal to
some positive integer $d_i$ and the other counter equals 0.
During the phase, the positive counter decreases, while the other
counter increases. The phase ends with the first counter 
containing 0 and the other counter containing $d_{i+1}$.
In the next phase, the modes of the counters are interchanged.
Thus, a sequence of configurations where the phases are changing
 will be of the form:
$$(q_1, d_1, 0), (q_2, 0, d_2), (q_3, d_3, 0), (q_4, 0, d_4), (q_5, d_5, 0), (q_6, 0, d_6), \dots$$
where the $q_i$'s  are states, with $q_1 = q_s$ (the
initial state), and $d_1, d_2, d_3, \ldots$ are positive
integers. The second component of
the configuration refers to the value of the first counter, and the third
component refers to the value of the second.
Also, notice that in going from state $q_i$ in phase $i$ to 
state $q_{i+1}$ in phase $i+1$, the 2-counter machine goes
through intermediate states.

For each $i$, there are 5 cases for the value of $d_{i+1}$
in terms of $d_i$:
$d_{i+1} = d_i, 2d_i, 3d_i, d_i/2, d_i/3$ 
(the division operation only occurs if the number is divisible
by 2 or 3, respectively).
The case applied is determined by $q_i$.
Hence, a function $h$ can be defined such that
if $q_i$
is the state at the start of phase $i$,
$d_{i+1} = h(q_i)d_i$, where $h(q_i)$ is one of
$1, 2, 3, 1/2, 1/3$.

Let $T$ be a 2-counter machine accepting a recursively enumerable
language that is not recursive.  Assume
that $q_1=q_s$ is the initial state, which is never re-entered,
and  if $T$ halts, it does so in a unique state $q_h$.
Let $Q$ be the states of $T$, and $1$ be a new symbol.

In what follows, $\alpha$ is any sequence of
the form $\#I_1 \#I_2 \#\cdots\# I_{2m}\#$ (thus we assume that
the length is even), where for each $i$, $1 \leq i \leq 2m$,
$I_i = q1^k$ for some $q \in Q$ and $k \ge 1$, represents
a possible configuration of $T$ at
the beginning of phase $i$, where $q$ is the state and
$k$ is the value of the first counter (resp., the second) if $i$
is odd (resp., even).

Define $L_0$ to be the set of all strings $\alpha$ such that 
\begin{enumerate}
\item $\alpha = \#I_1 \#I_2\# \cdots \#I_{2m}\#$;
\item $m \ge 1$;
\item for $1 \le j \le 2m-1$, $I_j \Rightarrow I_{j+1}$, i.e.,
if $T$ begins in configuration $I_j$, then after one phase,
$T$ is in configuration $I_{j+1}$ (i.e., $I_{j+1}$
is a valid successor of $I_j$);
\end{enumerate}

Then, the following was shown in \cite{EIMInsertion2015}.
\begin{lemma} \label{lem1}
$L_0$ is not in $\DCM \cup 2\DCM(1)$.
\end{lemma}

We will use this language exactly to show that taking either the suffix, infix, or outfix of a language in $\DCM(1,3), \DCM(2,1)$, or $2\DCM(1)$ can produce languages that are in neither $\DCM$ nor $2\DCM(1)$.

\begin{proposition}
\label{NonclosureSuffix}
There exists a language  $L \in \DCM(1,3)$ (respectively $L \in \DCM(2,1)$, and $L \in 2\DCM(1)$) such that $\suff(L) \not \in \DCM \cup 2\DCM(1)$,
$\inf(L) \not \in \DCM \cup 2\DCM(1)$, and $\outf(L) \not \in \DCM
\cup 2\DCM(1)$.
\end{proposition}
\begin{proof}
Let $L_0$ be the language defined above,
which is not in $\DCM \cup 2\DCM(1)$.  
Let $a, b$ be  new symbols.  Clearly, $bL_0b$
is also not in $\DCM \cup 2\DCM(1)$. A configuration of $T$ is any string of the form $q 1^k$, $q \in Q, k \geq 1$ (whether it appears in any computation or not). Then let
$$L = \{a^i b \# I_1 \# I_2 \#  \cdots \# I_{2m} \# b \mid \begin{array}[t]{l} I_1,  \ldots, I_{2m} \mbox{~are configurations of 2-counter machine~} T,  i \le 2m-1, \\  I_{i+1}
\mbox{~is not a valid successor of~} I_i \}.\end{array}$$
Clearly $L$ is in $\DCM(1,3)$, in $\DCM(2,1)$
(as $\DCM(1,3)$ is a subset of $\DCM(2,1)$ as mentioned in Section \ref{sec:prelims}), and in $2\DCM(1)$ (as $\DCM(1,3)$ is a subset of $2\DCM(1)$).

Let $L_1$ be $\suff(L)$.   Suppose  $L_1$ is in $\DCM$ (resp., $2\DCM(1)$).
Then $L_2 = \overline{L_1}$ is also in $\DCM$ (resp., $2\DCM(1)$) since both are closed under complement \cite{Ibarra1978,IbarraJiang}.

Let $R = \{b \# I_1 \# I_2 \cdots \# I_{2m} \# b  ~|~  I_1,
\ldots, I_{2m}$ are configurations of $T \}$.
Then since $R$ is regular, $L_3 = L_2 \cap R$
is in $\DCM$ (resp, $2\DCM(1)$) as both are closed under intersection with regular languages \cite{Ibarra1978,IbarraJiang}. We get a contradiction, since
$L_3 = bL_0b$.

Non-closure under infix and outfix can be shown similarly (for outfix, the intersection with $R$ enforces that only erasing of all of the $a$'s is considered).
\qed \end{proof}

This implies non-closure under left-quotient with regular languages, and this result also extends to the embedding operation, a generalization of outfix.

\begin{corollary}\label{leftquotR}
There exists $L \in \DCM(1,3)$ (respectively $L \in \DCM(2,1)$, and $L \in 2\DCM(1)$), and $R \in \REG$ such that 
$ \lquot{R}{L} \not \in \DCM \cup 2\DCM(1)$.
\end{corollary}

\begin{corollary}
Let $m>0$. Then there exists $L \in \DCM(1,3)$ (respectively $L \in \DCM(2,1)$, and $L \in 2\DCM(1)$)
such that $\emb(L, m) \not \in \DCM \cup 2\DCM(1)$.
\end{corollary}

The results of Proposition \ref{NonclosureSuffix} and Corollary \ref{leftquotR} are optimal for suffix and infix as these operations applied to $\DCM(1,1)$ are always in $\DCM$ by Corollary \ref{suffinfDCM} (and since $\DCM(1,2) = \DCM(1,1)$).
But whether the outfix and embedding operations applied to $\DCM(1,1)$ languages is always in $\DCM$ is an open question.

\section{Closure and Non-Closure for $\NPCM$, $\DPCM$, and $\DPDA$}
\label{sec:DPDA}

To start, we consider quotients of nondeterministic classes, then use these
results for contrast with deterministic classes.

\begin{proposition}
Let $\LL_1$ and $\LL_2$ be classes of languages where $\LL_1$ is
a full trio closed under intersection with languages in $\LL_2$, and
where $L \in \LL_2$ implies $\Sigma^* \# L, L\# \Sigma^* \in \LL_2$,
for an alphabet $\Sigma$ and new symbol $\#$.
Then $\LL_1$ is closed under left and right quotient with $\LL_2$.
\end{proposition}
\begin{proof}
For right quotient, let $L_1 \in \LL_1, L_2 \in \LL_2$.
If $L_1 \in \LL_1$, then using an inverse homomorphism (where the homomorphism is from $(\Sigma \cup \{\#\})^*$ to $\Sigma^*$ that erases $\#$ and fixes all other letters), and
intersection with the regular language $\Sigma^* \# \Sigma^*$, it follows that
$L_1'  = \{x \# y  \mid xy \in L_1\}$ is also in $\LL_1$.
Let $L_2' = \Sigma^* \# L_2 \in \LL_2$. Then $L = L_1' \cap L_2' \in \LL_1$.
Then, as every full trio is closed under gsm mappings, it follows
that $L_1 L_2^{-1} \in \LL_1$ by erasing everything starting at the $\#$
symbol.

Similarly with left quotient.
\qed
\end{proof}

\begin{corollary}\label{nondeterministicclosure}
$\NPCM$ ($\NCM$ respectively) is closed under left and right quotient with $\NCM$.
\end{corollary}
This follows since $\NPCM$ is a full trio closed under intersection with $\NCM$
\cite{Ibarra1978}, and $\NCM$ is closed under concatenation.

The question remains as to whether this is also true for deterministic machines instead. 
For machines with a stack, we have:
\begin{proposition} The right quotient of a $\DPDA(1)$ language
(i.e., deterministic linear context-free) with a
$\DCM(2,1)$ language is not necessarily an $\NPDA$ language.
\end{proposition}
\begin{proof}
Take the $\DPDA(1)$ language 
$L_1 = \{d^l c^k b^j a^i \# a^i b^j c^k d^l \mid i, j, k, l > 0\}$.
Take the $\DCM(2,1)$ language 
$L_2 = \{ a^i b^j c^i d^j \mid i,j >0\}$. This is clearly a
non-context-free language that is in $\DCM(2,1)$.
However, $L_1 L_2^{-1} = L_2^R$, which is also not context-free.
\qed
\end{proof}

Next we see that, in contrast to $\DCM$ and $\DPDA$, $\DPCM$ 
is closed under neither prefix nor suffix. Indeed, both $\DCM$ 
and $\DPDA$ are closed under prefix (and right quotient with regular sets), 
but not left quotient with regular sets. Yet combining their stores into one type of machine yields languages that are closed under neither.

\begin{proposition} $\DPCM$ is not closed under prefix or suffix.
\label{DPCMprefixsuffix}
\end{proposition}
\begin{proof}
Assume otherwise. Let $L$ be a language in $\NCM(1,1)$ that is not in $\DPCM$, which was shown to exist \cite{OscarNCMA2014journal}. Let $M$ be an $\NCM(1,1)$ machine 
accepting $L$. Let $T$ be a set of labels associated bijectively with transitions of $M$.
Consider the language 
$L' = \{ t_m \cdots t_1 \$ w \mid M \mbox{~accepts~} w \mbox{~via transitions~} t_1, \ldots , t_m \in T\}$.
This language is in $\DPCM$ since a $\DPCM$ machine $M'$ can be built that first pushes
$t_m \cdots t_1$, and then simulates $M$ deterministically on transitions $t_1, \ldots,  t_m$
while popping from the pushdown and reading $w$. Then $\suff(L') \cap \$ \Sigma^* = \$ L$,
a contradiction, as $\DPCM$ is clearly closed under left quotient with a single symbol.

Similarly for prefix, consider $L^R$, and create a machine $M^R$ accepting $L^R$, which is possible since
$\NCM(1,1)$ is closed under reversal. Then 
$L'' = \{ w \$ t_1 \cdots t_m \mid M^R \mbox{~accepts~} w^R \mbox{~via~} t_1, \ldots, t_m \in T\}$. This is also a $\DPCM$
language as one can construct a machine $M''$ that pushes $w$, then while popping $w^R$ letter-by-letter, simulates $M$ deterministically on transitions $t_1, \ldots, t_m$ on $w^R$.
Then $\pref(L'') \cap \Sigma^*\$ = L\$$, a contradiction, as $\DPCM$ is clearly closed under
right quotient with a single symbol.
\qed
\end{proof}

\begin{corollary} $\DPCM$ is not closed under right or left quotient with regular sets. 
\end{corollary}

Thus, the deterministic variant of Corollary \ref{nondeterministicclosure} gives non-closure.

The following is also evident from the proof of Proposition \ref{DPCMprefixsuffix}.

\begin{corollary} Every $\NCM$ language can be obtained by taking the right quotient (resp. left
quotient) of a $\DPCM$ language by a regular language.
\end{corollary}
The statement of this corollary cannot be weakened to taking the quotients
of a $\DPDA$ with a regular language, since $\DPDA$ is closed
under right quotient with regular languages \cite{harrison1978}.

Lastly, we will address the question of whether the left or right quotient of a $\DPDA$ language with a $\DCM$ language is always in $\DPCM$. 

\begin{proposition}
The right quotient (resp. left quotient) of a $\DPDA(1)$ language with a $\DCM(1,1)$ language can be outside $\DPCM$.
\end{proposition}
\begin{proof}
To start, it is known that there exists an $\NCM(1,1)$ language that is not in $\DPCM$ \cite{OscarNCMA2014journal}. Let
$L$ be such a language, and let $M$ be a $\NCM(1,1)$ machine accepting $L$. Then
$L^R$ is also an $\NCM(1,1)$ language, and let $M^R$ be an $\NCM(1,1)$ machine accepting it.
Let $T$ be a set of labels associated bijectively with transitions of $M^R$.

Then, we can create a $\DCM(1,1)$ machine $M'$ accepting words in 
$\#(\Sigma \cup T)^*$ such that after reading $\#$, $M'$ simulates
$M^R$ deterministically by reading a label $t \in T$ before
simulating $t$ deterministically. That is,
if $M'$ reads a letter $a \in \Sigma$, $M'$ stores it in a buffer (that can hold exactly one letter), and if $M'$ reads a letter
$t\in T$, $M'$ simulates $M^R$ on the letter $a$ in the buffer using transition $t$, completely
deterministically. Then if $t$ is a stay transition, the next letter must be in $T$, and the
buffer stays intact, whereas if $t$ is a right transition, then the buffer is cleared,
and the next letter must be in $\Sigma$. If the input is not of this form (for example, if there are two letters from
$\Sigma$ in a row, or a transition label representing a right move followed by another transition label), the machine crashes (cannot continue and does not accept). The first transition must also
be from the initial state, and the simulation must end in a final state.
It is clear then that if $h$ is a homomorphism
that erases letters of $T$ and fixes letters of $\Sigma$, then 
$h(L(M')) = L(M^R)$.

Then, consider the language 
$L_1 = \{ w \# x \mid w \in \Sigma^*, x \in (\Sigma \cup T)^*, h(x) = w^R\}$. Then $L_1 \in \DPDA(1)$.

Consider $L_2 = L_1 L(M')^{-1}$. Then $L_2 = \{ w \mid w \in \Sigma^*, \mbox{~there exists~} x \in (\Sigma \cup T)^* \mbox{~such that~} 
h(x) = w^R, \mbox{~and~} h(x) \in L(M^R)\}$. Hence,
$L_2 = \{ w \mid w \in \Sigma^*, w \in L(M)\} = L$, which is not in $\DPCM$.

Similarly for left quotient by using the $\DPDA(1)$ language 
$L_1 = \{ x \# w \mid w \in \Sigma^*, x \in (\Sigma \cup T)^*\}$.
\qed
\end{proof}

The following is also evident from the proof above.

\begin{corollary} Every $\NCM$ language can be obtained by taking the right quotient (resp. left
quotient) of a $\DPDA(1)$ language by a $\DCM$ language.
\end{corollary}
Again, this statement cannot be weakened to the right quotient of a $\DPDA$
with a regular language since $\DPDA$ languages are closed under right quotient with regular languages \cite{GinsburgDPDAs}.

\section{Right and Left Quotients of Regular Sets}
\label{sec:reg}

Let $\mathscr{F}$ be any family of languages (which
need not be recursively enumerable).   It is known that $\REG$ is closed under right quotient by languages in $\mathscr{F}$ \cite{HU}.  However, this closure need  not be
effective, as it will depend on the properties of $\mathscr{F}$.  The following 
is an interesting observation which connects decidability of
the emptiness problem to effectiveness of closure under right
quotient:

\begin{proposition} \label{reg1}
Let $\mathscr{F}$ be any family of languages which is effectively closed under intersection with regular sets and whose emptiness problem is decidable.
Then $\REG$ is effectively closed under both left and right quotient by languages in $\mathscr{F}$.
\end{proposition}
\begin{proof}
We will start with right quotient.

Let $L_1 \in \REG$ and $L_2$ be in $\mathscr{F}$.  Let $M$ be a DFA accepting $L_1$.
Let $q$ be a state of $M$, and 
$L_q = \{  y \mid M \mbox{~from initial state~} q \mbox{~accepts~} y \}$.
Let $Q' = \{ q \mid q \mbox{~is a state of~} M, L_q \cap L_2 \neq \emptyset \}$.
Since $\mathscr{F}$ is effectively closed under intersection with regular sets and has a
decidable emptiness problem, $Q'$ is computable.   Then a DFA $M'$ accepting
$L_1 L_2^{-1}$ can be obtained by  just making $Q'$ the set of accepting states
in $M$.

Next, for left quotient, let  $L_1$ be  in $\mathscr{F}$,
and $ L_2$ in $\REG$ be accepted by a DFA $M$ whose initial state is $q_0$.

Let $ L_q = \{ x \mid  M \mbox{~on input~} x \mbox{~ends in state~} q \}$.
Let $Q'  =  \{q \mid q \mbox{~is a state of~} M, L_q \cap L_1 \neq \emptyset \}$.  Then $Q'$
is computable, since  $\mathscr{F}$ is effectively closed under intersection with
regular sets and has a decidable emptiness problem.

We then construct an NFA (with $\lambda$-transitions) $M'$ to accept $L_1^{-1} L_2$ as follows:
$M'$ starting in state $q_0$ with input $y$, then
nondeterministically goes to a state $q$ in $Q'$ without reading
any input, and
then simulates the DFA $M$.

\qed
\end{proof}

\begin{corollary} \label{reg2}
$\REG$ is effectively closed under left and right quotient by languages in:
\begin{enumerate}
\item
the families of languages accepted by $\NPCM$ and $2\DCM(1)$ machines,
\item
the family of languages accepted by $\MPCA$s, $\TCA$s, $\QCA$s, and $\EPDA$s,
\item the families of ET0L and Indexed languages.
\end{enumerate}
\end{corollary}
\begin{proof}
These families are closed under intersection with regular sets. They have also
a decidable emptiness problem \cite{Harju2002278,Ah68,RS}. The family of ET0L languages and Indexed languages are discussed further in \cite{RS} and \cite{Ah68} respectively.
\qed
\end{proof}

\section{Closure for Bounded Languages}
\label{sec:bounded}

In this subsection, deletion operations applied to bounded and letter-bounded languages will be examined.
We will need the following corollary to Theorem \ref{unary}.
\begin{corollary} \label{unarycor}
Let $L \subseteq \#a^*\#$ be accepted by a
$2\NCM$.  Then $L$ is regular.
\end{corollary}

\begin{proposition} \label{bounded}
If $L$ is a bounded language accepted by either a finite-crossing $2\NCM$, an $\NPCM$ or a finite-crossing $2\DPCM$, then all of $\pref(L)$, $\suff(L)$, $\inf(L)$,
$\outf(L)$ can be accepted by a $\DCM$.
\end{proposition}
\begin{proof}
By Theorem \ref{WWW}, an $\NCM$ can be constructed that accepts $L$. Further, one
can construct $\NCM$'s accepting
$\pref(L), \suff(L), \inf(L), \outf(L)$ since one-way $\NCM$ is closed under prefix, suffix, infix and outfix. In addition, it is known that applying these operations on bounded languages produce only bounded languages. Thus, by another application of Theorem \ref{WWW}, the result can then be
converted to a $\DCM$.
\qed \end{proof}

The ``finite-crossing'' requirement in the proposition above is necessary:
\begin{proposition}
There exists a letter-bounded language $L$ accepted by 
a $2\DCM(1)$ machine which makes only one reversal on the counter
such that $\suff(L)$ (resp., $\inf(L)$, $\outf(L), \pref(L)$)
is not in $\DCM \cup 2\DCM(1)$.
\end{proposition}
\begin{proof}
Let $L = \{a^i \#b^j\# ~|~  i, j \ge 2, j$ is divisible by $i \}$.
Clearly, $L$ can be accepted by a $2\DCM(1)$ which makes only
one reversal on the counter.  If $\suff(L)$ is in $\DCM \cup 2\DCM(1)$,
then $L' = \suff(L) \cap \#b^+\#$ would be in $\DCM \cup 2\DCM(1)$.  
From Corollary \ref{unarycor}, we get a contradiction, since
$L'$ is not semilinear.
The other cases are shown similarly.
\qed \end{proof}

\section{Conclusions}

We investigated many different deletion operations  applied to languages accepted by one-way and two-way deterministic reversal-bounded multicounter machines, deterministic pushdown automata, and finite automata. The operations  include the prefix, suffix, infix, and outfix operations, as well as left and right quotient with languages from different families. Although it is frequently expected that language families defined from deterministic machines will not be closed under deletion operations,  we showed that $\DCM$ is closed under right quotient with languages from many different language families, such as the context-free languages. When starting with one-way deterministic machines with one counter that makes only one reversal ($\DCM(1,1)$), taking the left quotient with languages from many different language families, such as the context-free languages, yields only languages in $\DCM$ (by increasing the number of counters). It follows from
these results that the suffix or infix closure of $\DCM(1,1)$ languages are all in $\DCM$. These results are surprising given the nondeterministic behaviour of the deletion operations. However, for both $\DCM(1,3)$, or $\DCM(2,1)$, taking the left quotient (or even just the suffix operation) yields languages that can neither be accepted by deterministic reversal-bounded multicounter machines, nor by 2-way nondeterministic machines with one reversal-bounded counter ($2\DCM(1)$).  

Some interesting open questions remain.  For example,
is the outfix and embedding operations applied to $\DCM(1,1)$ languages always yield languages
in $\DCM$? Also, other deletion operations, such as schema for
parallel deletion \cite{parInsDel} have not yet been investigated applied
to languages accepted by deterministic machines.

\section*{Acknowledgements}
We thank the referees for helpful suggestions improving the presentation of the paper.


\begin{thebibliography}{10}
\expandafter\ifx\csname url\endcsname\relax
  \def\url#1{\texttt{#1}}\fi
\expandafter\ifx\csname urlprefix\endcsname\relax\def\urlprefix{URL }\fi
\expandafter\ifx\csname href\endcsname\relax
  \def\href#1#2{#2} \def\path#1{#1}\fi

\bibitem{TAMC2015}
J.~Eremondi, O.~H. Ibarra, I.~McQuillan, Deletion operations on deterministic
  families of automata, in: R.~Jain, S.~Jain, F.~Stephan (Eds.), Lecture Notes
  in Computer Science, Vol. 9076 of 12th Annual Conference on Theory and
  Applications of Models of Computation, TAMC 2015, Singapore, 2015, pp.
  388--399.

\bibitem{G75}
S.~Ginsburg, Algebraic and Automata-Theoretic Properties of Formal Languages,
  North-Holland Publishing Company, Amsterdam, 1975.

\bibitem{Baker1974}
B.~S. Baker, R.~V. Book, Reversal-bounded multipushdown machines, Journal of
  Computer and System Sciences 8~(3) (1974) 315--332.

\bibitem{Ibarra1978}
O.~H. Ibarra, Reversal-bounded multicounter machines and their decision
  problems, Journal of the ACM 25~(1) (1978) 116--133.

\bibitem{counterMembrane}
O.~H. Ibarra, On strong reversibility in {P Systems} and related problems,
  International Journal of Foundations of Computer Science 22~(01) (2011)
  7--14.

\bibitem{verification}
O.~H. Ibarra, J.~Su, Z.~Dang, T.~Bultan, R.~A. Kemmerer, Counter machines and
  verification problems, Theoretical Computer Science 289~(1) (2002) 165--189.

\bibitem{stringTransducers}
R.~Alur, J.~V. Deshmukh, Nondeterministic streaming string transducers, in:
  L.~Aceto, M.~Henzinger, J.~Sgall (Eds.), Automata, Languages and Programming,
  Vol. 6756 of Lecture Notes in Computer Science, Springer Berlin Heidelberg,
  2011, pp. 1--20.

\bibitem{modelChecking}
M.~Hague, A.~W. Lin, Model checking recursive programs with numeric data types,
  in: G.~Gopalakrishnan, S.~Qadeer (Eds.), Computer Aided Verification, Vol.
  6806 of Lecture Notes in Computer Science, Springer Berlin Heidelberg, 2011,
  pp. 743--759.

\bibitem{verificationDiophantine}
G.~Xie, Z.~Dang, O.~H. Ibarra, A solvable class of quadratic diophantine
  equations with applications to verification of infinite-state systems, in:
  J.~C. Baeten, J.~K. Lenstra, J.~Parrow, G.~J. Woeginger (Eds.), Automata,
  Languages and Programming, Vol. 2719 of Lecture Notes in Computer Science,
  Springer Berlin Heidelberg, 2003, pp. 668--680.

\bibitem{EIMInsertion2015}
J.~Eremondi, O.~Ibarra, I.~McQuillan, Insertion operations on deterministic
  reversal-bounded counter machines, in: A.~Dediu, E.~Formenti,
  C.~Mart\'in-Vide, B.~Truthe (Eds.), Lecture Notes in Computer Science, Vol.
  8977 of 9th International Conference on Language and Automata Theory and
  Applications, LATA 2015, Nice, France, 2015, pp. 200--211.

\bibitem{Chiniforooshan2012}
E.~Chiniforooshan, M.~Daley, O.~H. Ibarra, L.~Kari, S.~Seki, One-reversal
  counter machines and multihead automata: Revisited, Theoretical Computer
  Science 454 (2012) 81--87.

\bibitem{parInsDel}
L.~Kari, S.~Seki, Schema for parallel insertion and deletion: Revisited,
  International Journal of Foundations of Computer Science 22~(07) (2011)
  1655--1668.

\bibitem{HU}
J.~E. Hopcroft, J.~D. Ullman, Introduction to Automata Theory, Languages, and
  Computation, Addison-Wesley, Reading, MA, 1979.

\bibitem{boundedSemilin}
O.~H. Ibarra, S.~Seki, Characterizations of bounded semilinear languages by
  one-way and two-way deterministic machines, International Journal of
  Foundations of Computer Science 23~(6) (2012) 1291--1306.

\bibitem{JKT}
H.~J\"{u}rgensen, L.~Kari, G.~Thierrin, Morphisms preserving densities,
  International Journal of Computer Mathematics 78 (2001) 165--189.

\bibitem{IbarraJiang}
O.~H. Ibarra, T.~Jiang, N.~Tran, H.~Wang, New decidability results concerning
  two-way counter machines, SIAM J. Comput. 23~(1) (1995) 123--137.

\bibitem{Gurari1981220}
E.~M. Gurari, O.~H. Ibarra, The complexity of decision problems for finite-turn
  multicounter machines, Journal of Computer and System Sciences 22~(2) (1981)
  220--229.

\bibitem{GinsburgDPDAs}
S.~Ginsburg, S.~Greibach, Deterministic context free languages, Information and
  Control 9~(6) (1966) 620--648.

\bibitem{harrison1978}
M.~Harrison, Introduction to Formal Language Theory, Addison-Wesley series in
  computer science, Addison-Wesley Pub. Co., 1978.

\bibitem{Harju2002278}
T.~Harju, O.~Ibarra, J.~Karhumäki, A.~Salomaa, Some decision problems
  concerning semilinearity and commutation, Journal of Computer and System
  Sciences 65~(2) (2002) 278--294.

\bibitem{Vijayashanker:1987:STA:913947}
K.~Vijayashanker, A study of tree adjoining grammars, Ph.D. thesis,
  Philadelphia, PA, USA (1987).

\bibitem{Minsky}
M.~L. Minsky, Recursive unsolvability of {P}ost's problem of ``tag'' and other
  topics in theory of {T}uring {M}achines, Annals of Mathematics 74~(3) (1961)
  pp. 437--455.

\bibitem{OscarNCMA2014journal}
O.~H. Ibarra, Visibly pushdown automata and transducers with counters, {\it
  Fundamenta Informaticae}. To appear. (2016).

\bibitem{Ah68}
A.~V. Aho, Indexed grammars---an extension of context-free grammars, J.~ACM
  15~(4) (1968) 647--671.

\bibitem{RS}
G.~Rozenberg, A.~Salomaa, The Mathematical Theory of L Systems, Academic Press,
  Inc., New York, 1980.

\end{thebibliography}
\end{document}